\documentclass{amsart}
\usepackage{amssymb,mathrsfs,MnSymbol}
\usepackage{color}
\usepackage{hyperref}

\usepackage{hyperref}

\def\bC {\mathbf{C}}

\def\bR {\mathbf{R}}

\def\bV {\mathbf{V}}

\def\fH {\mathfrak{H}}
\def\fS {\mathfrak{S}}

\def\cC {\mathcal{C}}
\def\cD {\mathcal{D}}
\def\cE {\mathcal{E}}

\def\cH {\mathcal{H}}

\def\cK {\mathcal{K}}
\def\cL {\mathcal{L}}
\def\cM {\mathcal{M}}
\def\cN {\mathcal{N}}
\def\cP {\mathcal{P}}

\def\cR {\mathcal{R}}
\def\cS {\mathcal{S}}
\def\cT {\mathcal{T}}

\def\cV {\mathcal{V}}

\def\cX {\mathcal{X}}

\def\a {{\alpha}}
\def\b {{\beta}}
\def\g {{\gamma}}
\def\Ga {{\Gamma}}
\def\de {{\delta}}
\def\eps {{\epsilon}}

\def\l {{\lambda}}
\def\L {{\Lambda}}
\def\si {{\sigma}}

\def\om {{\omega}}

\def\d {{\partial}}
\def\grad {{\nabla}}
\def\Dlt {{\Delta}}

\def\rstr {{\big |}}

\def\la {\langle}
\def\ra {\rangle}
\def \La {\bigg\langle}
\def \Ra {\bigg\rangle}

\newcommand{\Div}{\operatorname{div}}

\newcommand{\Tr}{\operatorname{trace}}

\newcommand{\Lip}{\operatorname{Lip}}

\newcommand{\Ad}{\operatorname{\mathbf{ad}^*}}

\newcommand{\MKd}{\operatorname{dist_{MK,2}}}
\newcommand{\MKp}{\operatorname{dist_{MK,p}}}
\newcommand{\MK}{\operatorname{dist_{MK,1}}}
\newcommand{\Op}{\operatorname{OP}}

\def\hb {{\hbar}}

\newcommand{\ba}{\begin{aligned}}
\newcommand{\ea}{\end{aligned}}

\newcommand{\be}{\begin{equation}}
\newcommand{\ee}{\end{equation}}

\newcommand{\lb}{\label}

\newtheorem{Thm}{Theorem}[section]

\newtheorem{Cor}[Thm]{Corollary}
\newtheorem{Lem}[Thm]{Lemma}
\newtheorem{Def}[Thm]{Definition}

\begin{document}

\title{Mean-Field Limits in Statistical Dynamics}

\author[F. Golse]{Fran\c cois Golse}
\address[F.G.]{CMLS, \'Ecole polytechnique, 91128 Palaiseau Cedex, France}
\email{francois.golse@polytechnique.edu}

\begin{abstract}
These lectures notes are aimed at introducing the reader to some recent mathematical tools and results for the mean-field limit in statistical dynamics. As a warm-up, lecture 1 reviews the approach to the mean-field limit in classical mechanics 
following the ideas of W. Braun, K. Hepp and R.L. Dobrushin, based on the notions of phase space empirical measures, Klimontovich solutions and Monge-Kantorovich-Wasserstein distances between probability measures. Lecture 2 discusses
an analogue of the notion of Klimontovich solution in quantum dynamics, and explains how this notion appears in Pickl's method to handle the case of interaction potentials with a Coulomb type singularity at the origin. Finally, lecture 3 explains
how the mean-field and the classical limits can be taken jointly on quantum $N$-particle dynamics, leading to the Vlasov equation. These lectures are based on a series of joint works with C. Mouhot and T. Paul.
\end{abstract}

\maketitle


\section*{Introduction: What is a Mean-Field Dynamics?}


Consider a system of $N$ perfectly identical point particles, subject to pairwise interactions. We shall be concerned with the case where $N\gg 1$. For instance, in the case of an ideal gas, a volume of $22.4\cdot 10^{-3}\mathrm{m}^3$ 
contains $N_A\simeq 6.02\cdot 10^{23}$ (Avogrado's constant) gas molecules at the atmospheric pressure $101.3\,\mathrm{kPa}$ at the temperature of $0^\circ\mathrm{C}=273.1\mathrm{ K}$.

The evolution of such a system can be described either

\smallskip
\noindent
(a) by the system of motion equations (Newton's second law of motion) satisfied by each particle, or

\noindent
(b) by the motion equation for the ``typical particle'' driven by its collective interaction with all the other particles.

\smallskip
The description of such a large particle system following the procedure outlined in (b) is usually referred to as the ``mean-field approximation'' for the $N$-particle dynamics.

\smallskip
The advantages and drawbacks of each description can be summarized as follows:

\noindent
(a) is perfect in theory, but unfeasible in practice, since it involves observing initial data in a phase space of very high dimension ($6N$ components for the positions and momenta in the case $N$ point particles in the three-dimensional 
Euclidean space), not to mention the resolution of a coupled system of $6N$ differential equations;

\noindent
(b) is only an approximation, but is set on a phase space of relatively low (at least fixed and independent of $N$) dimension, specifically $6$ instead of $6N$.

\medskip
This obviously suggests the mathematical problem of justifying rigorously the mean-field approximation (b) starting from (a) as a first principle. Obviously, it would be desirable to obtain a convergence rate as the particle number $N\to\infty$,
in order to have an idea of the precision of the mean-field description.

\smallskip
There are numerous examples of mean-field equations in physics, such as

\smallskip
\noindent
(i) the Vlasov-Poisson or Vlasov-Maxwell systems used in the modeling of collisionless plasmas or ionized gases \cite{Glassey}, or

\noindent
(ii) the Hartree or Hartree-Fock equations used in quantum chemistry ab initio computations \cite{SzaboOstlund, ChadamGlassey}.

\smallskip
More recently, various mean-field theories have been proposed to describe the motion of living agents (such as the Vicsek, or Cucker-Smale models used to describe flocking or swarming). In the present lectures, we shall mostly focus on
well-known models studied in mathematical physics, but the reader should be aware that the ideas presented here could also be used in different contexts. For the same reason, our list of references will be very incomplete as regards these
relatively new applications of the mean-field theory.


\section*{Outline}


The purpose of these lecture notes is to introduce some new methods to handle the mean-field and classical limits in quantum mechanics. As a warm-up, we begin with Dobrushin's convergence rate estimate for the mean-field limit in classical
mechanics, assuming that the interaction force field is bounded and Lipschitz continuous (see \cite{Dobrushin}). One of the most innovative features in Dobrushin's bound is the use of optimal transport distances, specifically the Wasserstein
metric of exponent $1$ --- also referred to as the Monge-Kantorovich(-Rubinstein) metric: see chapter 7 in \cite{VillaniAMS} for a complete presentation of Wasserstein distances. Another key feature of Dobrushin's estimate is the notion of
Klimontovich solution (of the Vlasov equation). Specifically, the phase space empirical measure for a $N$-tuple of identical point particles is a weak solution of the Vlasov equation if and only if the $N$ particles evolve according to Newton's
second law of motion.

With Dobrushin's work as a motivation, the following objects have been defined in the past decade to handle the analogous problems in quantum mechanics:

\smallskip
\noindent
(a) an analogue of the Wasserstein distance of exponent $2$ for comparing a quantum density operator to a classical probability density, or two quantum density operators, and

\noindent
(b) a quantum analogue of Klimontovich solutions in quantum mechanics.

\smallskip
Lectures 2 and 3 discuss both the mean-field and the classical limits for the quantum $N$-body problem by means of the mathematical tools defined in (a)-(b). We have deliberately chosen to ignore the ``historic'' approach of the quantum
mean-field limit, involving BBGKY hierarchies. One of the drawbacks of the BBGKY hierarchy approach is the lack of uniformity in the semiclassical regime. Since this is one of our main interests in this course, we have decided not to include 
a presentation of the BBGKY hierarchy approach. The interested reader can find a rather detailed presentation of the fundamental mathematical techniques pertaining to the BBGKY hierarchy approach in \cite{FGTwente}.


\bigskip
\centerline{\textbf{Table of Contents}}

\smallskip
\noindent
Lecture 1: From Newton to Vlasov (mean-field limit in classical mechanics)

\smallskip
\noindent
Lecture 2: From Schr\"odinger to Hartree (mean-field limit in quantum mechanics)

\smallskip
\noindent
Lecture 3: Mean-field and classical limits in quantum mechanics


\section{Lecture 1: From Newton to Vlasov\\(Mean-Field Limit in Classical Mechanics)}


\subsection{The $N$-Body Problem in Classical Mechanics}


Consider a system of $N$ identical point particles of mass $m$ moving in the spatial domain $\bR^d$, subject to a pairwise interaction potential $V\equiv V(z)\in\bR$.

Let us write Newton's second law of motion for the $j$th particle:
$$
\left\{
\ba
{}&m\dot x_j=\xi_j\,,
\\
& &\qquad j=1,\ldots,N\,,
\\
&\dot\xi_j=\sum_{\genfrac{}{}{0pt}{3}{k=1}{k\not=j}}^N-\grad V(x_j-x_k)\,.
\ea
\right.
$$
Here $x_j\in\bR^d$ and $\xi_j\in\bR^d$ are respectively the position and momentum of the $j$th particle. The notation $\dot z(t)$ designates the time derivative $\frac{dz}{dt}(t)$, as usual in rational mechanics. Thus the first equation above 
is the kinematic definition of the momentum of the $j$th particle, while $-\grad V(x_j-x_k)$ is the force exerted by the $k$-th particle at the position $x_k$ on the $j$th particle at the position $x_j$.

\smallskip
For instance, the interaction potential $V$ could be the repulsive Coulomb potential between particles with the same electric charge (in the context of plasma physics), or the attractive gravitational potential (in the context of astronomy). In both
examples, the interaction potential $V\equiv V(z)$ is singular at $z=0$. 

\smallskip
For the sake of mathematical simplicity, we shall assume that the interaction potential $V$ satisfies the following conditions:

\smallskip
\noindent
\textbf{Assumptions on $V$}
$$
\ba
{}&\text{(H1)}\qquad&&V(z)=V(-z)&&\hbox{ for all }z\in\bR^d\,,
\\
&\text{(H2)}&&V\in C^1(\bR^d)\,,&&\hbox{ and }\grad V\in L^\infty(\bR^d)\cap\Lip(\bR^d)\,.
\ea
$$
Assumption (H1) corresponds to Newton's 3rd law: if the $j$th and the $k$th particles interact via the potential $V$, the force $-\grad V(x_j-x_k)$ exerted on the $j$th particle by the $k$th particle is the opposite of the force $-\grad V(x_k-x_j)$
exerted by the $j$th particle on the $k$th particle.

\smallskip
Henceforth, we systematically use the following notation to designate the $N$-tuple of positions and momenta of the $N$ particles:
$$
X_N:=(x_1,\ldots,x_N)\in\bR^{dN}\,,\quad\text{ and }\quad\Xi_N:=(\xi_1,\ldots,\xi_N)\in\bR^{dN}\,.
$$
By the Cauchy-Lipschitz theorem, for each $N$-tuple of initial positions $X_N^{in}\in\bR^{dN}$ and momenta $\Xi_N^{in}\in\bR^{dN}$, the differential system above has a unique solution 
\be\lb{DefPhi}
t\mapsto\Phi(t,X^{in}_N,\Xi^{in}_N)=(X_N(t),\Xi_N(t))
\ee
passing through $(X_N^{in},\Xi_N^{in})$ at time $t=0$ and defined for all $t\in\bR$.

\subsection{Mean Field Scaling}


Define scaled time $\hat t$, position $\hat x_j$ and momentum $\hat\xi_j$ for the $j$th particle by the formulas
$$
\hat t=t/N\,,\qquad\hat x_j(\hat t)=x_j(t)\,,\qquad\hat\xi_j(\hat t)=\xi_j(t)\,.
$$
In terms of these new dynamical quantities and time variable, the motion equations take the form
$$
\left\{
\ba
{}&mN\frac{d\hat x_j}{d\hat t}=\hat\xi_j\,,
\\	
&&\qquad j=1,\ldots,N\,,
\\
&N\frac{d\hat\xi_j}{d\hat t}=\sum_{\genfrac{}{}{0pt}{3}{k=1}{k\not=j}}^N-\grad V(\hat x_j-\hat x_k)\,,
\ea
\right.
$$
At this point, we assume that the total mass of the $N$-particle system is finite --- in other words that the mass of each particle is of order $1/N$:
$$
Nm=1\,.
$$

Therefore, after dropping hats on all variables, and reverting to the original notation $\dot z(t)$ to designate the time derivative of $z(t)$, our starting point is the scaled system of Newton's second laws of motion for each particle:
\be\lb{NbodyClass}
\left\{
\ba
{}&\dot x_j=\xi_j\,,
\\	
&&\qquad j=1,\ldots,N\,,
\\
&\dot\xi_j=\frac1N\sum_{\genfrac{}{}{0pt}{3}{k=1}{k\not=j}}^N-\grad V(x_j-x_k)\,.
\ea
\right.
\ee

\subsection{Vlasov Equation}


Our target equation, on the other hand, is the mean-field motion equation in classical mechanics, henceforth designated in general as the Vlasov equation (although the original Vlasov equation was written specifically for electrons 
in a plasma \cite{Vlasov}).

The unknown of the Vlasov equation is a single-particle phase space number density $f\equiv f(t,x,v)$, the number density of particles at the position $x\in\bR^d$ with momentum $\xi\in\bR^d$ at time $t$. For each $t$, the function
$(x,\xi)\mapsto f(t,x,\xi)$ is a probability density on $\bR^d\times\bR^d$. More generally, one could think of $f$ as a time-dependent Borel probability on the single-particle phase space $\bR^d\times\bR^d$, in which case  we shall
write it as $f(t,dxd\xi)$.

The Vlasov equation for $f$ takes the form
$$
(\d_t+\xi\cdot\grad_x)f-\grad_xV_f\cdot\grad_\xi f=0\,,\qquad x,\xi\in\bR^d\,,
$$
where $V_f\equiv V_f(t,x)$ is the \textit{mean-field potential} defined by the following formula
$$
V_f(t,x):=\iint_{\bR^d\times\bR^d}V(x-y)f(t,dyd\eta)=(V\star f(t,\cdot))(x)\,,\qquad x\in\bR^d\,.
$$

In other words, one can think of $(x,\xi)$ as the position and momentum of the \textit{typical particle}. The method of characteristics for the Vlasov equation tells us that
$$
f(t,x(t),\xi(t))=\text{Const.}\,,
$$
where $t\mapsto(x(t),\xi(t))$ is a solution of the differential system
$$
\left\{
\ba
{}&\dot x(t)=\xi(t)\,,
\\
&\dot\xi(t)=-\grad_xV_f(t,x(t))\,.
\ea
\right.
$$
In other words, the probability density $f$ is pushed forward by the flow of the differential system involving the mean-field potential, the self-consistent potential defined by $f$ itself.

\smallskip
Henceforth, the set of Borel probability measures on $\bR^n$ is denoted $\cP(\bR^n)$, and for each $k>0$, we denote by $\cP_k(\bR^n)$ the set of Borel probability measures with finite $k$-th order moment on $\bR^n$, i.e.
$$
\mu\in\cP_k(\bR^n)\iff\mu\in\cP(\bR^n)\text{ and }\int_{\bR^n}|x|^k\mu(dx)<\infty\,.
$$

The following existence and uniqueness result for the Vlasov equation is easy to prove (the proof is a simple variant of the proof of the Cauchy-Lipschitz theorem, and is left to the reader).

\begin{Thm}
For each initial data $f^{in}\in\cP_1(\bR^{2d})$, there exists a unique weak solution $f\in C([0,+\infty);w-\cP(\bR^{2d}))$ of the Vlasov equation such that $f\rstr_{t=0}=f^{in}$.
\end{Thm}

(Here, the notation $w-\cP(\bR^n)$ designates the set $\cP(\bR^n)$ equipped with its weak topology.)

\subsection{Empirical Measure and Klimontovich Solutions}


Consider a system of $N$ identical particles with positions and momenta $x_1,\xi_1,\ldots,x_n,\xi_N\in\bR^d$. The $N$-particle phase space empirical measure of this particle system is
$$
\mu_{(X_N,\Xi_N)}:=\frac1N\sum_{k=1}^N\de_{x_k,\xi_k}\in\cP(\bR^d\times\bR^d)\,.
$$
In other words, the $N$-particle phase space empirical measure is a symmetric function of the $N$-tuple of positions and momenta of the particles with values in the set of probability measures on the \textit{single-particle} phase space.

\begin{Thm}[Klimontovich]
The two conditions below are equivalent

\smallskip
\noindent
(a) the vector-valued function $\bR\ni t\mapsto(X_N,\Xi_N)(t)\in\bR^{2dN}$ is a solution of Newton's differential system of motion equations, and

\noindent
(b) the measure-valued function $\bR\ni t\mapsto\mu_{(X_N,\Xi_N)(t)}\in\cP(\bR^{2d})$ is weak solution of the Vlasov equation that is weakly continuous in time.
\end{Thm}

\begin{proof}
Since $V\in C^1(\bR^d)$ (by (H2)) and $V$ is even (by (H1)), then $\grad V$ is odd, so that $\grad V(0)=0$. Therefore
$$
\ba
\frac1N\sum_{\genfrac{}{}{0pt}{3}{k=1}{k\not=j}}^N\grad V(x_j(t)-x_k(t))&=\frac1N\sum_{k=1}^N\grad V(x_j(t)-x_k(t))
\\
&=\int_{\bR^{2d}}\grad V(x_j(t)-z)\mu_{(X_N,\Xi_N)(t)}(dzd\zeta)
\ea
$$
Thus Newton's second law of motion for the $j$th particle is the defining differential system for the characteristic curves of the Vlasov equation, localized at $(x_j(t),\xi_j(t))$. This observation and the method of characteristics immediately 
imply the announced result.
\end{proof}

\smallskip
Therefore, let $f^{in}$ be a probability density on $\bR^d\times\bR^d$, and choose a $N$-tuple of positions $X_N^{in}$ and momenta $\Xi_N^{in}$ so that
$$
\mu_{(X_N^{in},\Xi_N^{in})}\to f^{in}
$$
weakly in $\cP(\bR^{2d})$ as $N\to\infty$. Denoting by $\Phi(t,\cdot)$ the flow generated by the system of Newton's motion equations, Klimontovich's theorem reduces the question of whether
$$
\mu_{\Phi(t/N,X_N^{in},\Xi_N^{in})}\to f(t)
$$
weakly in $\cP(\bR^{2d})$ as $N\to\infty$ for each $t\ge 0$ to the continuous dependence of the solution of the Vlasov equation in terms of its initial data for the weak topology of probability measures. This has been observed by Braun
and Hepp in \cite{BraunHepp}.

\subsection{Wasserstein Distances}


In his remarkable paper \cite{Dobrushin}, Dobrushin has improved the weak compactness argument used in \cite{BraunHepp} (see also \cite{NeunWick} for a first approach to the same problem), and obtained a convergence rate formulated 
in terms of the Wasserstein distance of exponent one. Before stating Dobrushin's result, we first recall some basic facts on Wasserstein distances. The books \cite{VillaniAMS,AmbrosioGS} are excellent references for a more detailed study 
of these distances.

\begin{Def}[Couplings of probability measures]
For each pair $\mu,\nu\in\cP(\bR^n)$, a coupling of $\mu$ and $\nu$ is a probability measure $\pi\in\cP(\bR^n\times\bR^n)$ such that
$$
\iint_{\bR^n\times\bR^n}(\phi(x)+\psi(y))\si(dxdy)=\int_{\bR^n}\phi(x)\mu(dx)+\int_{\bR^n}\psi(y)\nu(dy)
$$
\end{Def}

The set of couplings of $\mu,\nu$ will be henceforth denoted $\cC(\mu,\nu)$; it is an easy exercise (left to the reader) to check that
$$
\mu,\nu\in\cP_p(\bR^n)\implies\cC(\mu,\nu)\subset\cP_p(\bR^n\times\bR^n)
$$
(In the literature on optimal transport, couplings are very often referred to as ``transport plans''.)

\begin{Def}[Monge-Kantorovich or Wasserstein Distances]
Let $p\in[1,\infty)$; for each $\mu,\nu\in\cP_p(\bR^n)$, the Monge-Kantorovich, or Wasserstein distance of exponent $p$ between $\mu$ and $\nu$ is
$$
\MKp(\mu,\nu)=\inf_{\pi\in\cC(\mu,\nu)}\left(\iint_{\bR^n\times\bR^n}|x-y|^p\pi(dxdy)\right)^{1/p}
$$
\end{Def}

A fundamental result on these distances is the following formula, which is a special case of \textit{Monge-Kantorovich duality} (see Theorem 1.3 in chapter 1 of \cite{VillaniAMS}):
$$
\MKp(\mu,\nu)^p=\sup_{\genfrac{}{}{0pt}{3}{\phi(x)+\psi(y)\le|x-y|^p}{\phi,\psi\in C_b(\bR^n)}}\left(\int_{\bR^n}\phi(x)\mu(dx)+\int_{\bR^n}\psi(x)\nu(dx)\right)\,.
$$
In particular
$$
\MK(\mu,\nu)=\sup_{\genfrac{}{}{0pt}{3}{\Lip(\phi)\le 1}{\phi\in C_b(\bR^n)}}\left|\int_{\bR^n}\phi(z)\mu(dz)-\int_{\bR^n}\phi(z)\nu(dz)\right|
$$
(this is the Kantorovich-Rubinstein theorem, stated as Theorem 1.14 in \cite{VillaniAMS}).

\subsection{Dobrushin's Inequality}


Let us return to the derivation of the Vlasov equation from the system of Newton's second law of motion written for each particle.

\begin{Thm}[Dobrushin's inequality]
Assume that $V$ satisfies (H1)-(H2). Let $f^{in}\in\cP_1(\bR^{2d})$, and let $f$ be the (weak) solution of the Vlasov equation with initial data $f^{in}$.  Let $t\mapsto(X_N,\Xi_N)(t)$ be the solution of Newton's scaled differential system \eqref{NBodyClass}
with initial data 
$(X^{in}_N,\Xi^{in}_N)$. Then, for each $t\ge 0$,
$$
\MK(\mu_{(X_N,\Xi_N)(t)},f(t,\cdot))\le\MK(\mu_{(X^{in}_N,\Xi^{in}_N)},f^{in})e^{t+2\Lip(\grad V)t}\,.
$$
\end{Thm}

The proof of Dobrushin's inequality is important to understand how the Monge-Kantorovich-Wasserstein distances can be used in the analysis of PDEs, and we shall present it in detail.

\begin{proof} Of course, there is nothing special with the choice of an empirical measure as one of the Vlasov solutions. Therefore, let $f^{in}$ and $g^{in}\in\cP_1(\bR^{2d})$, and let $f$ and $g$ be the solutions of the Vlasov equation
$$
\ba
{}&\d_tf+\{\tfrac12|\xi|^2+V_f(t,x),f\}\!=0\,,\qquad f\rstr_{t=0}=f^{in}\,,
\\
&\d_tg+\{\tfrac12|\eta|^2+V_g(t,y),\,g\}=0\,,\qquad g\rstr_{t=0}=g^{in}\,.
\ea
$$
We have used here the notion of \textit{Poisson bracket}, which is classical in rational mechanics, and whose definition is recalled below:
$$
\{\phi,\psi\}(z,\zeta):=\grad_\zeta\phi(z,\zeta)\cdot\grad_z\psi(z,\zeta)-\grad_z\phi(z,\zeta)\cdot\grad_\zeta\psi(z,\zeta)\,,
$$
for all $\phi,\psi\in C^1(\bR^n_z\times\bR^n_\zeta)$.

\smallskip
\noindent
\textit{Step 1: propagation of 1st order moment.} Since we seek to estimate the Wasserstein distance of exponent $1$ between two solutions of the Vlasov equation, we first prove that these solutions have finite first order moments for all times.

\begin{Lem} \lb{L-Lem1}
The weak solution $f\in C([0,+\infty),w-\cP(\bR^{2d}))$ satisfies
$$
M_1(t):=\int_{\bR^{2d}}\!(|x|\!+\!|\xi|)f(t,dxd\xi)\le M_1(0)e^{t(\max(1,\Lip(\grad V))+\Lip(\grad V))}
$$
for all $t\ge 0$.
\end{Lem}

\begin{proof}[Proof of Lemma \ref{L-Lem1}] Multiplying both sides of the Vlasov equation by $|x|+|\xi|$, and integrating by parts shows that
$$
\ba
\dot M_1(t)=\int_{\bR^{2d}}\!\{\tfrac12|\xi|^2\!\!+\!\!V_f(t,x),|x|\!+\!|\xi|\}f(t,dxd\xi)
\\
=\int_{\bR^{2d}}(\xi\!\cdot\!\tfrac{x}{|x|}\!-\!\grad V_f(t,x)\!\cdot\!\tfrac{\xi}{|\xi|})f(t,dxd\xi)
\\
\le\int_{\bR^{2d}}(|\xi|+|\grad V_f(t,x)|)f(t,dxd\xi)&\,.
\ea
$$
Observe that
$$
\ba
|\grad_xV_f(t,x)-\grad_xV_f(t,0)|
\\
\le\int_{\bR^{2d}}|\grad V(x-z)-\grad V(-z))|f(t,dzd\zeta)\le\Lip(\grad V)|x|&\,,
\ea
$$
since $f(t,\cdot,\cdot)$ is a probability measure, while
$$
\ba
\grad V(0)=0\implies|\grad_xV_f(t,0)|\le\int_{\bR^{2d}}|\grad V(-z)|f(t,dzd\zeta)
\\
\le\Lip(\grad V)\int_{\bR^{2d}}|z|f(t,dzd\zeta)\le\Lip(\grad V)M_1(t)&\,.
\ea
$$
Hence
$$
\ba
\dot M_1(t)\le\int_{\bR^{2d}}(|\xi|+\Lip(\grad V)(|x|+M_1(t)))f(t,dxd\xi)
\\
\le(\max(1,\Lip(\grad V))+\Lip(\grad V))M_1(t)&\,,
\ea
$$
and the sought inequality follows from Gronwall's lemma.
\end{proof}

\smallskip
\noindent
\textit{Step 2: propagation of couplings.} Let $h^{in}\in\cC(f^{in},g^{in})$; we seek to construct an element of $\cC(f(t),g(t))$ for all $t\ge 0$. One way of doing this (by no means the only one) is provided by the following lemma.

\begin{Lem}\lb{L-Lem2}
Let $h$ be the weak solution of the Liouville equation in $\bR^{2d}_{x,\xi}\times\bR^{2d}_{y,\eta}$
$$
\d_th+\{\tfrac12|\xi|^2+\tfrac12|\eta|^2+V_f(t,x)+V_g(t,y),h\}=0\,,\qquad h\rstr_{t=0}=h^{in}\,,
$$
where $h^{in}\in\cC(f^{in},g^{in})$. The Poisson bracket used here corresponds to choosing $n=2d$, with $z=(x,y)$ and $\zeta=(\xi,\eta)$. Then
$$
h^{in}\in\cC(f^{in},g^{in})\implies h(t)\in\cC(f(t),g(t))\quad\text{ for each }t\ge 0\,.
$$
\end{Lem}

\begin{proof}[Proof of Lemma \ref{L-Lem2}] For each $\phi\in C^1_c(\bR^{2d})$, one has
$$
\ba
\frac{d}{dt}\int_{\bR^{4d}}\phi(x,\xi)h(t,dxd\xi dyd\eta)
\\
=\int_{\bR^{4d}}\{\tfrac12|\xi|^2+\tfrac12|\eta|^2+V_f(t,x)+V_g(t,y),\phi(x,\xi)\}h(t,dxd\xi dyd\eta)
\\
=\int_{\bR^{4d}}\{\tfrac12|\xi|^2+V_f(t,x),\phi(x,\xi)\}h(t,dxd\xi dyd\eta)&\,.
\ea
$$
By uniqueness of the solution of the Liouville equation with  initial data $f^{in}$ with Hamiltonian 
$$
\tfrac12|\xi|^2+V_f(t,x)\,,
$$
this implies that the first marginal of $h(t)$ is 
$$
\int_{\bR^{2d}}h(t)dyd\eta=f(t)\,.
$$
\end{proof}

\smallskip
\noindent
\textit{Step 3: growth of the Monge-Kantorovich distance.} With $h$ defined in Step 2, consider the quantity
$$
D(t):=\int_{\bR^{4d}}(|x-y|+|\xi-\eta|)h(t,dxd\xi dyd\eta)\,.
$$
Then
$$
\dot D(t)=\int_{\bR^{4d}}B(t,x,\xi,y,\eta)h(t,dxd\xi dyd\eta)\,,
$$
with
$$
\ba
B(t,x,\xi,y,\eta)\!=&\{\tfrac12|\xi|^2\!+\!\tfrac12|\eta|^2\!+\!V_f(t,x)\!+\!V_g(t,y),|x\!-\!y|\!+\!|\xi\!-\!\eta|\}
\\
=&(\xi-\eta)\cdot\tfrac{x-y}{|x-y|}\!-\!(\grad_xV_f(t,x)\!-\!\grad_yV_g(t,y))\cdot\tfrac{\xi-\eta}{|\xi-\eta|}
\\
\le&|\xi-\eta|+|\grad_xV_f(t,x)\!-\!\grad_yV_g(t,y)|\,.
\ea
$$
Now
$$
\ba
|\grad_xV_f(t,x)\!-\!\grad_yV_g(t,y)|
\\
\le\int_{\bR^{2d}}|\grad V(x-z)-\grad V(y-z)|f(t,dzd\zeta)
\\
+\left|\int_{\bR^{2d}}\grad V(y-z)f(t,dzd\zeta)-\int_{\bR^{2d}}\grad V(y-z)g(t,dzd\zeta)\right|
\\
\le\Lip(\grad V)|x-y|+\Lip(\grad V)\MK(f(t),g(t))&\,,
\ea
$$
since $f(t,\cdot,\ cdot)$ is a probability measure, so that
$$
B(t,x,\xi,y,\eta)\le|\xi-\eta|+\Lip(\grad V)|x-y|+\Lip(\grad V)\MK(f(t),g(t))\,.
$$
Hence
$$
\ba
\dot D(t)\le&\int_{\bR^{4d}}(|\xi-\eta|+\Lip(\grad V)|x-y|)h(t,dxd\xi dyd\eta)
\\
&+\Lip(\grad V)\MK(f(t),g(t)))
\\
\le&\max(1,\Lip(\grad V))D(t)+\Lip(\grad V)\MK(f(t),g(t)))\,.
\ea
$$
By Lemma \ref{L-Lem2}, for each $t\ge 0$, one has $h(t)\in\cC(f(t),g(t))$, and hence
$$
\MK(f(t),g(t))\le D(t)\,,
$$
so that
$$
\dot D(t)\le(\max(1,\Lip(\grad V))+\Lip(\grad V))D(t)\,.
$$
On the other hand, by Gronwall's lemma,
$$
\MK(f(t),g(t)))\le D(t)\le D(0)e^{t(\max(1,\Lip(\grad V))+\Lip(\grad V))}\,.
$$
Minimizing the last right hand side in $h^{in}\in\cC(f^{in},g^{in})$ implies that
$$
\MK(f(t),g(t)))\le \MK(f^{in},g^{in})e^{t(\max(1,\Lip(\grad V))+\Lip(\grad V))}\,.
$$
\end{proof}

\subsection{Applications of Dobrushin's Inequality to the Mean-Field Limit}


By Theorem 7.12 in chapter 7 of \cite{VillaniAMS}, the Monge-Kantorovich distance $\MK$ metrizes the weak topology of probability measures on $\cP_1(\bR^{2d})$ --- see \cite{VillaniAMS} for a more precise statement, including
the convergence of linearly growing test functions at infinity. By a density argument, pick a sequence of initial position and momenta $(X_N^{in},\Xi_N^{in})$ such that
$$
\mu_{X_N^{in},\Xi_N^{in}}\to f^{in}\text{ weakly in }\cP(\bR^{2d})
$$
and
$$
\frac1N\sum_{j=1}^N(|x_{j,N}^{in}|+|\xi_{j,N}^{in}|)\to\int_{\bR^{2d}}(|x|+|\xi|)f^{in}(dxd\xi)
$$
as $N\to\infty$. By Theorem 7.12 in chapter 7 of \cite{VillaniAMS}, 
$$
\MK(\mu_{X_N^{in},\Xi_N^{in}},f^{in})\to 0\quad\text{ as }N\to\infty\,,
$$
and Dobrushin's inequality implies that
$$
\MK(\mu_{X_N(t),\Xi_N(t)},f(t))\to 0\quad\text{ for each }t\ge 0\text{ as }N\to\infty\,.
$$
This justifies the mean-field limit in classical mechanics for identical point particles interacting via a potential $V$ satisfying assumptions (H1)-(H2).

\smallskip
However, one can improve this result and obtain a quantitative statement with a convergence rate, provided that one can estimate the speed of convergence of the initial empirical measure $\mu_{X_N^{in},\Xi_N^{in}}$ to $f^{in}$.
This can be done by using quantitative variants of the strong law of large numbers. The following bound has been obtained by Fournier and Guillin \cite{FournierGuillin}.

\begin{Thm}[Fournier-Guillin]
Assume that $f^{in}\in\cP_q(\bR^{2d})$ with $1<q\not=\tfrac{2d}{2d-1}$ and $d\ge 3$. Then
$$
\int_{\bR^{2dN}}\!\MK(\mu_{(X^{in}_N,\Xi^{in}_N)},f^{in})\prod_{j=1}^Nf^{in}(dx_jd\xi_j)\!\le\!CM_q^{\frac1q}\left(\frac1{N^{\frac1q}}\!+\!\frac1{N^{1-\frac1q}}\right)\,,
$$
where
$$
M_q:=\iint_{\bR^d\times\bR^d}(|x|+|\xi|)^qf^{in}(x,\xi)dxd\xi<\infty\,.
$$
\end{Thm}

\smallskip
Using both the Dobrushin inequality and the Fournier-Guillin bound leads to the following statement on the mean-field limit in classical mechanics.

\begin{Cor}
Let $f^{in}\in\cP_q(\bR^{2d})$ with $1<q\not=\tfrac{2d}{2d-1}$ and $d\ge 3$, and let $V$ satisfy (H1)-(H2). Let $f$ be the solution of the Vlasov equation with initial data $f\rstr_{t=0}=f^{in}$, and let $\Phi(t,\cdot)$ be the one-parameter flow 
\eqref{DefPhi}. Then
$$
\ba
\int_{\bR^{2dN}}\MK(\mu_{\Phi(t/N,X^{in}_N,\Xi^{in}_N)},f(t,\cdot))\prod_{j=1}^Nf^{in}(dx_jd\xi_j)
\\
\le
CM_q^{1/q}e^{t+2\Lip(\grad V)t}\left(N^{-\frac1q}+N^{-(1-\frac1q)}\right)&\,.
\ea
$$
\end{Cor}

\smallskip
There are several limitations in the derivation of the Vlasov equation from the classical $N$-body dynamics which are inherent to the Dobrushin approach.

\smallskip
First and foremost, Dobrushin's method seems limited to Lipschitz continuous interaction forces. This is a serious drawback, since it rules out such physically interesting interactions as the Coulomb, or screened Coulomb, or Yukawa repulsive
potentials, as well as the Newton's gravitational potential. The Dobrushin approach can be modified to treat singular force fields that are less singular at the origin than the Coulomb or gravitational forces (see \cite{HaurayJabin1,HaurayJabin2}).
Another possibility is to start from a mollified interaction at the origin, removing the regularization parameter as $N\to\infty$ (see \cite{Pickl-Lazarovici,Lazarovici}). 

\smallskip
Another potentially annoying peculiarity of the Dobrushin approach to the justification of the mean-field limit in classical mechanics is that it uses mathematical objects which seem particular to the classical setting, and whose extension to quantum
dynamics seems far from obvious (phase space empirical measures, individual particle trajectories, Klimontovich solutions and so on).


\section{Lecture 2: From Schr\"odinger to Hartree\\ (Mean-Field Limit in Quantum Mechanics)}


\subsection{The Quantum $N$-Body Dynamics}


The state at time $t$ of an $N$-particle system in quantum mechanics is described by its wave function
$$
\Psi_N\equiv\Psi_N(t,x_1,\ldots,x_N)\in\bC\,,
$$
assumed to satisfy the normalization condition
$$
\int_{\bR^{dN}}|\Psi_N(t,X_N)|^2dX_N=1\,,\qquad\text{ with }X_N:=(x_1,\ldots,x_N)\,.
$$
We recall that $|\Psi_N(t,X_N)|^2dX_N$ should be thought of as the joint probability of finding particle $1$ in an infinitesimal neighborhood of volume $dx_1$ centered at position $x_1\in\bR^d$, particle $2$ in an infinitesimal neighborhood 
of volume $dx_2$ centered at position $x_2\in\bR^d$, \dots, and particle $N$ in an infinitesimal neighborhood of volume $dx_N$ centered at position $x_N\in\bR^d$.

The wave function is governed by the Schr\"odinger equation
$$
i\hb\d_t\Psi_N=\cH_N\Psi_n\,,\qquad\Psi_N\rstr_{t=0}=\Psi_N^{in}\,,
$$
with quantum $N$-body Hamiltonian
$$
\cH_N:=\sum_{j=1}^N\underbrace{-\tfrac1{2m}\hb^2\Dlt_{x_j}}_{\text{kinetic energy}}+\sum_{1\le j<k\le N}\underbrace{V(x_j-x_k)}_{\text{potential energy}}
$$

The Schr\"odinger equation is the quantum analogue of the system of Newton's motion equation in classical mechanics presented in lecture 1. While the existence of the classical dynamics rests on the Cauchy-Lipschitz theorem, which
requires the interaction force field $\grad V$ to be Lipschitz continuous, the existence of the quantum $N$-particle dynamics follows from the following fundamental result due to T. Kato \cite{Kato51}.

\begin{Thm} [Kato] If $d=3$, and if for some $R>0$, 
$$
V\rstr_{B(0,R)}\in L^2(B(0,R))\qquad\hbox{ while }V\rstr_{\bR^3\setminus B(0,R)}\in L^\infty(\bR^3\setminus B(0,R))\leqno{\mathrm{(H3)}}
$$
for each $N>1$ and each $m,\hbar>0$, the operator $\cH_N$, which is defined as a linear map from $\cS(\bR^{dN})$ to $L^2(\bR^{dN})$, has an unbounded self-adjoint extension on $L^2(\bR^{3N})$.
\end{Thm}

\smallskip
In particular, this extension generates a unitary group $e^{-it\cH_N}$ on $L^2(\bR^{3N})$ (by Stone's theorem).

Notice that the condition (H3) on the potential used in Kato's theorem to define the quantum $N$-body dynamics is much weaker than the condition (H2) used to define the classical $N$-body dynamics via the Cauchy-Lipschitz theorem.
In particular, Kato's condition (H3) is satisfied by the repulsive Coulomb potential between identical charged particles, a special case of considerable interest in atomic physics.

\subsection{The Quantum Mean-Field Dynamics}


Exactly as in the context of classical dynamics, we assume that the total mass of our $N$ particle system is of order $1$ as $N\to\infty$, i.e. $Nm=1$, and consider the dynamics in time $\hat t=t/N$, i.e. the unitary group
$$
\exp\left(-\frac{it}{\hb N}\cH_N\right)=\exp\left(-\frac{it}{\hb}\widehat{\cH_N}\right)
$$
where
$$
\widehat\cH_N:=\frac{\cH_N}{N}=\sum_{j=1}^N-\tfrac1{2}\hb^2\Dlt_{x_j}+\frac1N\sum_{1\le j<k\le N}V(x_j-x_k)\,.
$$
Henceforth we consider as our starting point the Schr\"odinger equation defined by $\widehat{\cH_N}$. For notational simplicity, we also drop the hat on the rescaled Hamiltonian $\widehat{\cH_N}$.

By analogy with the classical problem studied in lecture 1, it is natural to replace the $N$-body potential acting on the $j$th particle, viz.
$$
\frac1N\sum_{\genfrac{}{}{0pt}{3}{k=1}{k\not=j}}^NV(x_j-x_k)
$$
with its mean-field approximation, which is the convolution of $V$ with the single particle density function at time $t$, i.e. $|\psi(t,x)|^2$, where $\psi\equiv\psi(t,x)$ is the wave function of the typical particle in the $N$-particle
system under consideration. In other words, the mean-field potential is
$$
V_\psi(t,x):=V\star|\psi(t,\cdot)|^2(x)=\int_{\bR^d}V(x-y)|\psi(t,y)|^2dy\,.
$$
The corresponding mean-field Hamiltonian is the operator
$$
-\tfrac12\hb^2\Dlt_x+V_\psi(t,x)\,,
$$
and the single particle wave function $\psi$ describing the quantum state of the typical particle satisfies the \textit{time-dependent Hartree} (TDH) equation
$$
i\hb\d_t\psi(t,x)=-\tfrac12\hb^2\Dlt_x\psi(t,x)+V_\psi(t,x)\psi(t,x)\,,\qquad\psi\rstr_{t=0}=\psi^{in}\,.
$$

The following (formal) computations are left to the reader as (easy) exercises: the conservation of particle number is 
$$
\frac{d}{dt}\int_{\bR^d}|\psi(t,x)|^2dx=0\,,
$$
leading to the propagation of the normalization condition:
$$
\|\psi(t,\cdot)\|_{L^2(\bR^d)}=\|\psi^{in}\|_{L^2(\bR^d)}=1\,,\qquad t\ge 0\,.
$$
The conservation of energy takes the form
$$
\frac{d}{dt}\left(\int_{\bR^d}\tfrac12\hb^2|\grad_x\psi(t,x)|^2dx+\tfrac12\iint_{\bR^d\times\bR^d}V(x-y)|\psi(t,x)|^2|\psi(t,y)|^2dxdy\right)=0\,,
$$
so that, if $\psi^{in}$ has finite energy and if $V\ge 0$ on $\bR^d$, then the solution of the TDH equation satisfies
$$
\psi\in L^\infty((0,\infty);H^1(\bR^d))\,.
$$

\subsection{Reduced Density Operators}


Henceforth assume that $\Psi_N(t,\cdot)$ is a symmetric function of the position variables for each particle. This symmetry assumption\footnote{The assumption that the particles considered here are bosons is not necessary for most of 
the mathematical results considered in this lecture. However, the mean-field scaling assumed in this and the next lecture is specific to bosons, and differs from the one used in the case of fermions (particles with half-integral spin): see
Remark (7) below, at the end of lecture 2. Except for the mean-field scaling, most of the results discussed in this lecture hold for a system of indistinguishable particles, bosons or fermions.} corresponds to assuming that the particles are 
\textit{bosons} (i.e. have integral spin: see \S 61 in \cite{LL3}). Thus, for all $\si\in\fS_N$, for a.e. $X_N\in\bR^{Nd}$ and all $t\ge 0$, one has
$$
U_\si\Psi_N(t,X_N):=\Psi_N(t,x_{\si^{-1}(1)},\ldots,x_{\si^{-1}(N)})=\Psi_N(t,X_N)\,.
$$
To the wave function $\Psi_N$, one associates the $N$-body density operator $R_N(t)$ on $\fH$ with integral kernel
$$
r_N(t,X_N,Y_N):=\Psi_N(t,X_N)\overline{\Psi_N(t,Y_N)}\,.
$$
Obviously, $R_N(t)$ is the orthogonal projection on the line $\bC\Psi_N$ in the Hilbert space $\fH_N:=L^2(\bR^{dN})$, since $\Psi_N(t,\cdot)$ is assumed to satisfy the normalization condition
$$
\|\Psi_N(t,\cdot)\|_{L^2(\bR^{dN})}=1\,.
$$

For each $k=1,\ldots,N-1$, the $k$-particle reduced density operator is the integral operator denoted $R_{N:k}(t)$ on $\fH_k=L^2(\bR^{dk})$ with integral kernel
$$
r_{N:k}(t,X_k,Y_k):=\int_{\bR^{d(N-k)}}r_N(t,X_k,Z_{k,N},Y_k,Z_{k,N})dZ_{k,N}\,,
$$
where we have denoted
$$
Z_{k,N}:=(z_{k+1},\ldots,z_N)\,.
$$

In the sequel, we shall systematically use \textit{Dirac's bra-ket notation}: each function $\Phi_N\in L^2(\bR^{dN})$ defines a vector of $\fH_N$ denoted $|\Phi_N\ra$ (a ket, involving only a closing bracket). Similarly, to each 
$\Psi_N\in L^2(\bR^{dN})$, one associates the linear functional 
$$
\fH_N\ni\Phi_N\mapsto\int_{\bR^{dN}}\overline{\Psi_N(X_N)}\Phi_N(X_N)dX_N=\la\Psi_N|\Phi_N\ra\in\bC\,.
$$
Since the function $\Phi_N$ can be viewed as the vector $|\Phi_N\ra$ of $\fH_N$, the notation $\la\cdot|\cdot\ra$ for the inner product in the Hilbert space $\fH_N$ makes it natural to denote this linear functional as $\la\Psi_N|$ 
(a bra, involving only an opening bracket).

\subsection{Quantum Klimontovich Solutions}


After these preliminaries, we arrive at the main task in this lecture, namely defining the quantum analogue of the notions of empirical measure and Klimontovich solution in classical mechanics. The material in this section is
taken from \cite{FGTPEmpirical}. We first define these notions, and then explain why these definitions are natural by analogy with the classical setting.

As above, we set $\fH:=L^2(\bR^d;\bC)$ (the single-particle Hilbert space in space dimension $d$), and for each integer $N\ge 1$ (the particle number) $\fH_N:=L^2(\bR^{dN};\bC)$ (the $N$-particle Hilbert space). For each
integer $k=1,\ldots,N$, set
$$
J_k:\,\cL(\fH)\ni A\mapsto J_kA:=\underbrace{I_\fH\otimes\ldots\otimes A\otimes\ldots\otimes I_\fH}_{\hbox{$A$ on the $k$-th variable}}\in\cL(\fH_N)\,.
$$
With this notation, we first define the quantum analogue of the notion of empirical measure.

\begin{Def} For each $N>1$, we set
$$
\cM^{in}_N:=\frac1N\sum_{k=1}^NJ_k\in\cL(\cL(\fH),\cL_s(\fH_N))\,,
$$
where $\cL(E,F)$ designates the set of continuous linear maps from the Banach space $E$ to the Banach space $F$, while 
$$
\cL_s(\fH_N):=\{T\in\cL(\fH_N)\text{ s.t. }U_\si TU_\si^*=T\text{ for all }\si\in\fS_N\}\,.
$$
\end{Def}

Why this is indeed a natural quantum analogue of the notion of phase space empirical measure in classical mechanics may require some explanation. 

In quantum mechanics, one associates to physical quantities (such as position, momentum, energy, angular momentum \ldots) self-adjoint operators with pure point spectrum and a complete orthonormal set of eigenfunctions 
in the Hilbert space of the system considered. Such operators are called ``observables'' in the language of quantum mechanics, and the ``expected value'' of the physical quantity corresponding to the operator $A=A^*$ for the 
system in the state associated to the wave function $\Psi$ is 
$$
\la\Psi|A|\Psi\ra\in\bR\,.
$$
(Indeed, let $(\phi_j)_{j\ge 1}$ be a complete orthonormal system such that $A\psi_j=a_j\psi_j$ with $k\not=j\implies a_k\not=a_j$; the probability that the physical quantity associated with the observable $A$ takes the
value $a_j$ on the system in the state associated to the wave function $\Psi$ is $p_j:=|\la\phi_j|\Psi\ra|^2$ (assuming of course that $\|\Psi\|=1$). Thus
$$
\la\Psi|A|\Psi\ra=\sum_{j\ge 1}a_jp_j\,,
$$
which confirms the interpretation of $\la\Psi|A|\Psi\ra$ as a mathematical expectation.) See \S\S 1-7 in chapter V of \cite{Messiah} for more detail on these important notions.

Thus, if $A=A^*\in\cL(\fH)$ is a single-particle observable, $\cM^{in}_NA$ is an $N$-particle observable that is invariant under permutation of the particle labels, i.e. 
$$
U_\si(\cM^{in}_NA)U_\si^*=\cM^{in}_NA\qquad\text{ for each }\si\in\fS_N\,.
$$

The corresponding statement involving the phase space empirical measure is as follows: let $f\equiv f(z,\zeta)$ be a real-valued, continuous bounded function defined on the single-particle phase space $\bR^d\times\bR^d$.
One can think of this function as representing the phase space density of some physical quantity (such as the kinetic energy $f(z,\zeta)=|\zeta|^2/2m$ for a particle with momentum $\zeta\in\bR^d$ and mass $m>0$). Then
$$
\int_{\bR^d\times\bR^d}f(z,\zeta)\mu_{X_N^{in},\Xi_N^{in}}(dzd\zeta)=\La\frac1N\sum_{j=1}^N\de_{z_j,\zeta_j},f\Ra=\frac1N\sum_{j=1}^Nf(x_j,\xi_j)\,,
$$
and the phase space empirical measure $\mu_{X_N^{in},\Xi_N^{in}}$ can be thought of as the ``integral kernel'' of the linear map
$$
C_b(\bR^d\times\bR^d)\ni f\mapsto F_N\equiv F_N(X_N,\Xi_N):=\frac1N\sum_{j=1}^Nf(x_j,\xi_j)\in C_b(\bR^{dN}\times\bR^{dN})\,.
$$
Besides $F_N$ is obviously symmetric in the variables $(x_k,\xi_k)$, in other words, is invariant under perturbations of the particle labels).

\smallskip
We recall from lecture 1 that Klimontovich solutions of the Vlasov equation are phase space empirical measures of the form 
$$
\mu_{\Phi(t/N,X_N^{in},\Xi_N^{in})}\,,
$$
where $\Phi(t,\cdot)$ is the Hamiltonian flow of
$$
\sum_{j=1}^N\tfrac1{2m}|\xi_j|^2+\sum_{1\le j<k\le N}V(x_j-x_k)
$$
defined on $\bR^{dN}\times\bR^{dN}$, assuming that the total mass of the system is $Nm=1$. Its quantum analogue is defined as follows.

\begin{Def}
The quantum Klimontovich solution for an $N$-particle system governed by the dynamics associated to the quantum $N$-particle Hamiltonian $\cH_N$ is the time-dependent element $\cM_N(t)$ of $\cL(\cL(\fH),\cL_s(\fH_N))$
defined by the formula
$$
\cM_N(t)A:=e^{it\cH_N/\hb}(\cM_N^{in}A)e^{-it\cH_N/\hb}\,,\qquad A\in\cL(\fH)\,,\,\,t\in\bR\,.
$$
\end{Def}

\smallskip
Next we study a characteristic property of $\cM_N(t)$ --- which could indeed serve as an alternative definition of $\cM_N(t)$. 

\begin{Lem}\lb{L-CharPropMN(t)}
For each $\phi\in\fH$ and each $\Psi_N^{in}\in\fH_N$ satisfying the bosonic symmetry 
$$
\Psi_N^{in}=U_\si\Psi_N^{in}\in\fH_N\qquad\text{ for all }\si\in\fS_N\,,
$$
one has
$$
\la\Psi_N^{in}|\cM_N(t)(|\phi\ra\la\phi|)|\Psi_N^{in}\ra_{\fH_N}=\la\phi|R_{N:1}(t)|\phi\ra_\fH\,,\qquad\text{ for all }t\ge 0\,.
$$
\end{Lem}

In other words, $\cM_N(t)$ is the adjoint of the linear map
$$
|\Psi_N^{in}\ra\la\Psi_N^{in}|\mapsto R_{N:1}(t)
$$
where $R_{N:1}(t)$ is the single-particle reduced density operator associated to 
$$
\Psi_N(t)=e^{-it\cH_N/\hb}\Psi_N^{in}\,.
$$

\begin{proof} 
Since 
$$
\Psi_N^{in}=U_\si\Psi_N^{in}\implies\Psi_N(t)=U_\si\Psi_N(t)
$$ 
for all $\si\in\fS_N$, then
$$
\ba
\la\Psi_N^{in}|\cM_N(t)(|\phi\ra\la\phi|)|\Psi_N^{in}\ra_{\fH_N}\!=\!\la e^{-\frac{i}\eps\cH_N}\Psi_N^{in}|\cM_N^{in}(|\phi\ra\la\phi|)|e^{-\frac{i}\eps\cH_N}\Psi_N^{in}\ra_{\fH_N}
\\
=\!\La\Psi_N(t)\Bigg|\frac1N\sum_{k=1}^NJ_k(|\phi\ra\la\phi|)\Bigg|\Psi_N(t)\Ra_{\fH_N}\!\!\!\!=\!\la\Psi_N(t)|J_1(|\phi\ra\la\phi|)|\Psi_N(t)\ra_{\fH_N}
\\
=\int_{\bR^{d(N-1)}}|\la\Psi_N(t,\cdot,X_{2,N})|\phi\ra|^2dX_{2,N}=\la\phi|R_{N:1}(t)|\phi\ra_\fH&\,.
\ea
$$
The symmetry of $\Psi_N(t)$ has been used to prove the third equality.
\end{proof}

\smallskip
There is a similar property for the Klimontovich solution in the classical setting. Let $F_n^{in}\equiv F_N^{in}(X_N,\Xi_N)$ be a probability density on $\bR^{2dN}$ satisfying the symmetry
$$
F_N^{in}(X_N,\Xi_N)=F_N^{in}(x_{\si^{-1}(1)},\ldots,x_{\si^{-1}(N)},\xi_{\si^{-1}(1)},\ldots,\xi_{\si^{-1}(N)})\,,
$$
for all $X_N,\Xi_N\in\bR^{dN}$ and all $\si\in\fS_N$. Let $\Phi(t,\cdot)$ be the Hamiltonian flow of
$$
\sum_{j=1}^N\tfrac1{2m}|\xi_j|^2+\sum_{1\le j<k\le N}V(x_j-x_k)\,.
$$
Then
$$
\int_{\bR^{2dN}}\mu_{\Phi(t/N,X_N^{in},\Xi_N^{in})}F_N^{in}(X^{in}_N,\Xi^{in}_N)dX^{in}_Nd\Xi^{in}_N=F_{N:1}(t,\cdot)\,,
$$
where
$$
F_{N:1}(t,x_1,\xi_1)\!\!:=\!\!\!\int_{\bR^{2d(N-1)}}\!F_N^{in}(\Phi(-\tfrac{t}{N},x_1,\ldots,x_N,\xi_1,\ldots,\xi_N))dx_2\ldots dx_Nd\xi_2\ldots d\xi_N
$$
is the first marginal of the $N$-particle distribution function at time $t$. (See formula (32) in \cite{FGMouRicci}.) Equivalently
$$
\int_{\bR^{2d}}\phi(z,\zeta)F_{N:1}(t,z,\zeta)dzd\zeta=\int_{\bR^{2dN}}\la\mu_{\Phi(t/N,X_N,\Xi_N)},\phi\ra F_N^{in}(X^{in}_N,\Xi^{in}_N)dX^{in}_Nd\Xi^{in}_N
$$
for all $t\in\bR$, and each $\phi\in C_b(\bR^{2d})$. This is analogous to the formula in Lemma \ref{L-CharPropMN(t)} in the classical setting.

\subsection{An Equation for Quantum Klimontovich Solutions}


In this section, we shall assume for the sake of simplicity that the potential $V\in C_0(\bR^d)$ satisfies (H1) and that its Fourier transform $\hat V$ satisfies
$$
\hat V\in L^1(\bR^d)\,.\leqno{\mathrm{(H4)}}
$$
For each $\om\in\bR^d$, we denote by $E_\om$ the multiplication operator on $\fH=L^2(\bR^d)$ defined by the formula
$$
(E_\om\psi)(x)=e^{i\om\cdot x}\psi(x)\,,\qquad\psi\in\fH\,.
$$
Obviously 
$$
E_\om^*=E_{-\om}=E_\om^{-1}\quad\text{ for each }\om\in\bR^d\,.
$$

Next, for each linear map $\L:\,\cL(\fH)\to\cL(\fH_N)$, each unbounded operator $H$ on the single-particle Hilbert space $\fH$ and each $A\in\cL(\fH)$ satisfying
$$
[H,A]\!:=HA\!-\!AH\in\cL(\fH)\,,
$$
we set
$$
(\Ad(H)\L)A:=-\L([H,A])\,.
$$
(We have chosen this notation by analogy with the co-adjoint representation of a Lie algebra: if $\fH=\bC^n$, if $A,H\in\cL(\fH)=M_n(\bC)$, which is the Lie algebra of the group $GL_n(\bC)$, and if $L\in M_n(\bC)^*$ is a linear 
functional on $M_n(\bC)$, the coadjoint representation $M_n(\bC)\ni H\mapsto\Ad(H)\in\cL(M_n(\bC)^*)$ is defined by
$$
M_n(\bC)^*\ni L\mapsto\Ad(H)L\in M_n(\bC)^*
$$
where
$$
\Ad(H)L:\,M_n(\bC)\ni A\mapsto -\la L,[H,A]\ra\in\bC\,.
$$
The main difference with the situation considered here is that $\L$ is an operator-valued linear map, instead of being a linear functional, so that there is no duality in our setting. The term $\Ad(H)$ is used here only for lack of 
a more convenient notation.)

While the notation $\Ad(H)$ is used for the kinetic energy in the quantum Hamiltonian, we need another notation for the interaction term, i.e. the potential energy in the quantum Hamiltonian. For all $\L_1,\L_2\in\cL(\cL(\fH),\cL(\fH_N))$,
we define the linear map $\cC[V,\L_1,\L_2]\in\cL(\cL(\fH),\cL(\fH_N))$ by the formula
$$
\cC[V,\L_1,\L_2]A\!:=\!\!\!\int_{\bR^d}((\L_1E^*_\om)\L_2(E_\om A)\!-\!\L_2(AE_\om)(\L_1E^*_\om))\hat V(\om)\tfrac{d\om}{(2\pi)^d}
$$
for all $A\in\cL(\fH)$. Some care should be exercised with the definition of the integral in the right hand side of this formula, since it takes its values in $\cL(\fH_N)$, which is not separable. 

Let $\cX$ be a Banach space, with topological dual denoted $\cX'$. The weak-* topology on $\cX'$ is the topology defined by the family of seminorms $\ell\mapsto|\la\ell,x\ra_{X',X}|$ as $x$ runs through $X$. Let $f:\,\bR^d\to\cX'$ 
be weakly-* continuous and bounded (for the norm topology) on $\bR^d$. Let $m$ be a (bounded) complex Borel measure on $\bR^d$; then the linear functional
$$
\cX\ni\phi\mapsto\la L_{f,m},\phi\ra:=\int_{\bR^d}\la f(\om),\phi\ra m(d\om)\in\bC
$$
is continuous with norm
$$
\|L_{f,m}\|_{\cX'}\le\sup_{\om\in\bR^d}\|f(\om)\|_{\cX'}\|m\|_{TV}\,.
$$
This defines the integral
$$
\int_{\bR^d}f(\om)m(d\om):=L_{f,m}\in\cX'\,.
$$
In the case of the integral in the definition of $\cC[V,\L_1,\L_2]A$, the Banach space $\cX$ is $\cL^1(\fH_N)$ (the space of trace-class operators on $\fH_N$) and its topological dual is $\cL(\fH_N)$, with duality defined by the trace:
$$
\la B\,T\ra_{\cL^(\fH_N),\cL^1(\fH_N)}=\Tr_{\fH_N}(BT)\,.
$$
The weak-* topology is the ultraweak topology on $\cL(\fH_N)$. This construction defines the integral of the bounded, ultraweakly continuous function $f:\,\bR^d\to\cL(\fH_N)$ with respect to the complex Borel measure $m$ on $\bR^d$ 
$$
\int_{\bR^d}f(\om)m(d\om)\in\cL(\fH_N)
$$
as an element of $\cL(\fH_N)$ identified with the continuous linear functional 
$$
\cL^1(\fH_N)\ni T\mapsto\int_{\bR^d}\Tr_{\fH_N}(f(\om)T)m(d\om)\in\bC
$$
by the formula
$$
\Tr_{\fH_N}\left(\left(\int_{\bR^d}f(\om)m(d\om)\right)T\right)=\int_{\bR^d}\Tr_{\fH_N}(f(\om)T)m(d\om)\,.
$$
(For the reader familiar with these notions, $\cC[V,\L_1,\L_2]A$ is defined by duality, as a (Gelfand-)Pettis integral instead of a Bochner integral.)

\smallskip
With these notations, we can formulate our first main result, i.e. the governing equation satisfied by the time-dependent linear map $\cM_N(t)$.

\begin{Thm} [\cite{FGTPEmpirical}]\lb{T-qKlim}
Let $V$ be a real-valued function satisfying (H1) and (H4). Then
$$
i\hb\d_t\cM_N(t)=\Ad(-\tfrac12\hb^2\Dlt)\cM_N(t)-\cC[V,\cM_N(t),\cM_N(t)]\,.
$$
\end{Thm}

\begin{proof} Start from Lemma \ref{L-CharPropMN(t)}:
$$
\ba
i\hb\d_t\la\Psi_N^{in}|\cM_N(t)(|\phi\ra\la\phi|)|\Psi_N^{in}\ra=&i\hb\d_t\la\Psi_N^{in}|e^{+it\cH_N/\hb}\cM_N^{in}(|\phi\ra\la\phi|)e^{-it\cH_N/\hb}|\Psi_N^{in}\ra
\\
=&i\hb\d_t\la e^{-it\cH_N/\hb}\Psi_N^{in}|\cM_N^{in}(|\phi\ra\la\phi|)|e^{-it\cH_N/\hb}\Psi_N^{in}\ra
\\
=&-\la\cH_Ne^{-it\cH_N/\hb}\Psi_N^{in}|\cM_N^{in}(|\phi\ra\la\phi|)|e^{-it\cH_N/\hb}\Psi_N^{in}\ra
\\
&+\la e^{-it\cH_N/\hb}\Psi_N^{in}|\cM_N^{in}(|\phi\ra\la\phi|)|\cH_Ne^{-it\cH_N/\hb}\Psi_N^{in}\ra
\\
=&-\la\Psi_N^{in}|e^{+it\cH_N/\hb}[\cH_N,\cM_N^{in}(|\phi\ra\la\phi|)]e^{-it\cH_N/\hb}|\Psi_N^{in}\ra\,.
\ea
$$
Split the Hamiltonian $\cH_N$ as $\cH_N=\cK_N+\cV_N$ where $\cK_N$ is the kinetic energy and $\cV_N$ the potential energy, i.e.
$$
\cK_N:=\sum_{k=1}^nJ_k(-\tfrac12\hb^2\Dlt)\,,\qquad\cV_N=:\tfrac1N\sum_{1\le k<l\le N}V_{kl}\,,
$$
where 
$$
V_{kl}\Psi_N(x_1,\ldots,x_N):=V(x_k-x_l)\Psi_N(x_1,\ldots,x_N)\,.
$$

First
$$
\ba
{}[\cK_N,\cM_N^{in}(|\phi\ra\la\phi|)]=&\left[\sum_{k=1}^NJ_k(-\tfrac12\hb^2\Dlt),\tfrac1N\sum_{l=1}^NJ_l(|\phi\ra\la\phi|)\right]
\\
=&\tfrac1N\sum_{l=1}^NJ_l([-\tfrac12\hb^2\Dlt,|\phi\ra\la\phi|])
\\
=&\cM_N^{in}([-\tfrac12\hb^2\Dlt,|\phi\ra\la\phi|])\,,
\ea
$$
so that
$$
\ba
e^{+it\cH_N/\hb}[\cK_N,\cM_N^{in}(|\phi\ra\la\phi|)]e^{-it\cH_N/\hb}=&e^{+it\cH_N/\hb}\cM_N^{in}([-\tfrac12\hb^2\Dlt,|\phi\ra\la\phi|])e^{-it\cH_N/\hb}
\\
=&\cM_N(t)([-\tfrac12\hb^2\Dlt,|\phi\ra\la\phi|])
\\
=&-\Ad(-\tfrac12\hb^2\Dlt)\cM_N(t)(|\phi\ra\la\phi|)\,.
\ea
$$

Next, we use (H4) and the Fourier inversion formula to write
$$
V_{kl}=\tfrac1{(2\pi)^d}\int_{\bR^d}\hat V(\om)J_kE_\om J_lE_\om^*d\om\,.
$$
Thus
$$
\ba
{}[\cV_N,\cM_N^{in}(|\phi\ra\la\phi|)]=&\tfrac1{(2\pi)^d}\int_{\bR^d}\hat V(\om)\left[\tfrac1N\sum_{1\le k<l\le N}J_kE_\om J_lE_\om^*,\tfrac1N\sum_{m=1}^NJ_m|\phi\ra\la\phi|\right]d\om
\\
=&\tfrac1{(2\pi)^d}\int_{\bR^d}\hat V(\om)\tfrac1{N^2}\sum_{1\le k<l\le N}J_k[E_\om,|\phi\ra\la\phi|]J_lE_\om^*d\om
\\
&+\tfrac1{(2\pi)^d}\int_{\bR^d}\hat V(\om)\tfrac1{N^2}\sum_{1\le k<l\le N}J_kE_\om J_l[E^*_\om,|\phi\ra\la\phi|]d\om
\\
=&\tfrac1{(2\pi)^d}\int_{\bR^d}\hat V(\om)\tfrac1{N^2}\sum_{1\le k\not=l\le N}J_k[E_\om,|\phi\ra\la\phi|]J_lE_\om^*d\om\,,
\ea
$$
where the last equality follows from (H1), which implies that $\hat V(\om)=\hat V(-\om)$ for all $\om\in\bR^d$, and from the fact that $J_kA$ commutes with $J_lB$ for $k\not=l$. This last formula can be recast as
$$
\ba
{}[\cV_N,\cM_N^{in}(|\phi\ra\la\phi|)]
\\
=\tfrac1{(2\pi)^d}\int_{\bR^d}\hat V(\om)\tfrac1{N^2}\sum_{1\le k\not=l\le N}\left(J_lE_\om^*J_k(E_\om|\phi\ra\la\phi|)-J_k(|\phi\ra\la\phi|E_\om)J_lE_\om^*\right)d\om
\\
=\tfrac1{(2\pi)^d}\int_{\bR^d}\hat V(\om)\tfrac1{N^2}\sum_{1\le k,l\le N}\left(J_lE_\om^*J_k(E_\om|\phi\ra\la\phi|)-J_k(|\phi\ra\la\phi|E_\om)J_lE_\om^*\right)d\om
\\
=\tfrac1{(2\pi)^d}\int_{\bR^d}\hat V(\om)\left(\cM_N^{in}(E_\om^*)\cM_N^{in}(E_\om|\phi\ra\la\phi|)-\cM_N^{in}(|\phi\ra\la\phi|E_\om)\cM_N^{in}(E_\om^*)\right)d\om
\\
=\cC[V,\cM_N^{in},\cM_N^{in}](|\phi\ra\la\phi|)&\,.
\ea
$$
In the left-hand side of the second equality, observe that the operators $J_lE_\om^*$ and $J_k(E_\om|\phi\ra\la\phi|)$ or $J_k(|\phi\ra\la\phi|E_\om)$ obviously commute (by definition of $J_k$) since $k\not=l$. In the right hand side of the second equality,
the operators $J_lE_\om^*$ and $J_k(E_\om|\phi\ra\la\phi|)$ or $J_k(|\phi\ra\la\phi|E_\om)$ do not commute in general unless $k\not=l$, but one easily checks that
$$
\ba
J_kE_\om^*J_k(E_\om|\phi\ra\la\phi|)=J_k(E_\om^*E_\om|\phi\ra\la\phi|)=&J_k(|\phi\ra\la\phi|)
\\
=&J_k(|\phi\ra\la\phi|E_\om E_\om^*)=J_k(|\phi\ra\la\phi|E_\om)J_k(E_\om^*)\,.
\ea
$$
This explains why the second equality holds true. The remaining equalities being obvious, we are left with the task of computing
$$
\ba
e^{+it\cH_N/\hb}[\cV_N,\cM_N^{in}(|\phi\ra\la\phi|)]e^{-it\cH_N/\hb}=e^{+it\cH_N/\hb}\cC[V,\cM_N^{in},\cM_N^{in}](|\phi\ra\la\phi|)e^{-it\cH_N/\hb}
\\
=\tfrac1{(2\pi)^d}\int_{\bR^d}\hat V(\om)\big(e^{+it\cH_N/\hb}\cM_N^{in}(E_\om^*)e^{-it\cH_N/\hb}e^{+it\cH_N/\hb}\cM_N^{in}(E_\om|\phi\ra\la\phi|)e^{-it\cH_N/\hb}
\\
-e^{+it\cH_N/\hb}\cM_N^{in}(|\phi\ra\la\phi|E_\om)e^{-it\cH_N/\hb}e^{+it\cH_N/\hb}\cM_N^{in}(E_\om^*)e^{+it\cH_N/\hb}\big)d\om
\\
=\!\!\int_{\bR^d}\!\!\!\hat V(\om)\big(\cM_N(t)(E_\om^*)\cM_N(t)(E_\om|\phi\ra\la\phi|)\!-\!\cM_N(t)(|\phi\ra\la\phi|E_\om)\cM_N(t)(E_\om^*)\big)\tfrac{d\om}{(2\pi)^d}
\\
=\cC[V,\cM_N(t),\cM_N(t)](|\phi\ra\la\phi|)&.
\ea
$$

Summarizing, we have seen that
$$
\ba
i\hb\d_t\la\Psi_N^{in}|\cM_N(t)(|\phi\ra\la\phi|)|\Psi_N^{in}\ra
\\
=-\la\Psi_N^{in}|e^{+it\cH_N/\hb}[\cK_N+\cV_N,\cM_N^{in}(|\phi\ra\la\phi|)]e^{-it\cH_N/\hb}|\Psi_N^{in}\ra
\\
=\la\Psi_N^{in}|\Ad(-\tfrac12\hb^2\Dlt)\cM_N(t)(|\phi\ra\la\phi|)-\cC[V,\cM_N(t),\cM_N(t)](|\phi\ra\la\phi|)|\Psi_N^{in}\ra
\ea
$$
for each $\Psi_N^{in}$ in the domain of $\cH_N$ and each $\phi$ in the domain of $\Dlt$, i.e. for each $\phi\in H^2(\bR^d)$. Observing that
$$
A=A^*\in\cL(\fH)\implies\cM_N(t)A=e^{it\cH_N/\hb}\left(\tfrac1N\sum_{k=1}^NJ_kA\right)e^{-it\cH_N/\hb}=(\cM_N(t)A)^*
$$
shows that
$$
\left(\Ad(-\tfrac12\hb^2\Dlt)\cM_N(t)(|\phi\ra\la\phi|)\right)^*=-\Ad(-\tfrac12\hb^2\Dlt)\cM_N(t)(|\phi\ra\la\phi|)\,,
$$
while
$$
\cC[V,\cM_N(t),\cM_N(t)](|\phi\ra\la\phi|)^*=-\cC[V,\cM_N(t),\cM_N(t)](|\phi\ra\la\phi|)\,.
$$
Thus the operator
$$
\ba
\cT:=\d_t\cM_N(t)(|\phi\ra\la\phi|)+&\tfrac{i}\hb\Ad(-\tfrac12\hb^2\Dlt)\cM_N(t)(|\phi\ra\la\phi|)
\\
-&\tfrac{i}\hb\cC[V,\cM_N(t),\cM_N(t)](|\phi\ra\la\phi|)=\cT^*\in\cL(\fH_N)
\ea
$$
satisfies
$$
\la\Psi_N^{in}|\cT|\Psi_N^{in}\ra=0
$$
for each $\Psi_N^{in}$ in the domain of $\cH_N$. By polarization, we conclude that $\cT=0$.
\end{proof}

\smallskip
There is a marked difference between the Klimontovich theorem in classical mechanics, and the previous theorem. Indeed, at first sight, the equation satisfied by $\cM_N(t)$ differs from the TDH equation, which is the quantum mean-field
equation analogous to the Vlasov equation in classical mechanics. However, this first impression is quite misleading, as shown by our next theorem. Before stating this theorem, we need to explain better how a Vlasov solution $f(t,dzd\zeta)$ 
can be compared to the time-dependent phase space empirical measure 
$$
\mu_{\Phi(t/N,X_N^{in},\Xi_N^{in})}(dzd\zeta)=\frac1N\sum_{j=1}^N\de_{x_j(t),\xi_j(t)}(dzd\zeta)
$$
associated to an $N$-particle system (where we recall that $\Phi$ is the Hamiltonian flow generated by the classical Hamiltonian
$$
\sum_{j=1}^N\tfrac1{2m}|\xi_j|^2+\sum_{1\le j<k\le N}V(x_j-x_k)
$$
with $Nm=1$ and with interaction potential $V$ satisfying (H1)-(H2)). Of course, both $f(t,dzd\zeta)$ and $\mu_{\Phi(t/N,X_N^{in},\Xi_N^{in})}(dzd\zeta)$ are Borel probability measures on the phase space $\bR^d\times\bR^d$, but the Klimontovich solution depends 
on the initial $N$-tuple of positions and momenta of the particle system, i.e. $(X_N^{in},\Xi_N^{in})$, whereas the Vlasov solution $f(t,dzd\zeta)$ is obviously independent of these initial coordinates.

\smallskip
Therefore, in the quantum setting, we must consider objects analogous to the Klimontovich solution $\cM_N(t)$ but ``independent'' of the initial $N$-particle coordinates --- in other words, a constant function of these parameters. 

\begin{Def}\lb{D-ChaoMor}
Let $\psi\in C(\bR;L^2(\bR^d))$ be a time-dependent wave function such that $\|\psi(t,\cdot)\|_\fH=1$ for all $t\in\bR$. The element $\cR_\psi(t)\in\cL(\cL(\fH),\cL_s(\fH_N))$ of the form
$$
\cR_\psi(t)A:=\la\psi(t,\cdot)|A|\psi(t,\cdot)\ra\,I_{\fH_N}\,,\qquad A\in\cL(\fH)\,,
$$
is called the ``chaotic morphism'' associated to the wave function $\psi$.
\end{Def}

In the classical setting, the analogous object is
$$
C_b(\bR^{2d})\ni\phi\mapsto\left(\int_{\bR^{2d}}\phi(z,\zeta)f(t,dzd\zeta)\right)1\in C_b(\bR^{2dN})\,,
$$
where $1$ is the constant function of $(X_N^{in},\Xi_N^{in})$ in $C_b(\bR^{2dN})$. (The terminology ``chaotic'' comes from the law of large numbers: if $(x_j,\xi_j)$ are mutually independent random phase space coordinates distributed
according to $f(t,dzd\zeta)$, the phase space empirical measure $\mu_{(X_N,\Xi_N)}\to f(t,dzd\zeta)$ weakly as $N\to\infty$, almost surely in the sequence $(x_j,\xi_j)_{j\ge 1}$. The limit as $N\to\infty$ of the phase space empirical
measure is in particular independent of, or constant in the sequence of phase space coordinates $(x_j,\xi_j)_{j\ge 1}$.)

\smallskip
The quantum analogue of the Klimontovich theorem is obtained by inserting a chaotic morphism in the equation satisfied by $\cM_N(t)$ presented in the preceding theorem.

\begin{Thm}\lb{T-MFh=1}
Let $V$ be a real-valued potential satisfying assumptions (H1)-(H4). Let $\psi\equiv\psi(t,x)$ be a solution of the Hartree equation
$$
\left\{
\ba
{}&i\hb\d_t\psi=-\tfrac12\hb^2\Dlt_x\psi+(V\star|\psi(t,\cdot)|^2)\psi\,,\qquad x\in\bR^d\,,
\\
&\psi\rstr_{t=0}=\psi^{in}\,,
\ea
\right.
$$
with initial data $\psi^{in}\in H^1(\bR^d)$ satisfying the normalization $\|\psi^{in}\|_{L^2(\bR^d)}=1$. Then the chaotic morphism $t\mapsto\cR_\psi(t)$ is a solution of the equation
$$
i\hb\d_t\cR(t)=\Ad(-\tfrac12\hb^2\Dlt)\cR(t)-\cC[V,\cR(t),\cR(t)]\,.
$$
\end{Thm}

In other words, the time-dependent Hartree equation (TDH) is a special case of the equation governing the evolution of quantum Klimontovich solutions $\cM_N(t)$ obtained in Theorem \ref{T-qKlim}.

\begin{proof}
Since $\cR(t)A=\la\psi(t,\cdot)|A|\psi(t,\cdot)\ra_{\fH}\,I_{\fH_N}$ then
$$
\cR(t)E^*_\om=\widehat{|\psi|^2}(t,\om)\,I_{\fH_N}\,,\qquad\om\in\bR^d\,,
$$
is the Fourier transform of the density function $x\mapsto|\psi(t,x)|^2$ associated to the Hartree solution. Setting $A:=|\phi\ra\la\phi|$ with $\phi\in H^2(\bR^d)$, one has therefore
$$
\ba
\cC[V,\cR_\psi(t),\cR_\psi(t)]A
\\
=\int_{\bR^d}((\cR_\psi(t)E^*_\om)\cR_\psi(t)(E_\om A)\!-\!\cR_\psi(t)(AE_\om)(\cR_\psi(t)E^*_\om))\hat V(\om)\tfrac{d\om}{(2\pi)^d}
\\
=\int_{\bR^d}\cR_\psi(t)([E_\om,A])\hat V(\om)\widehat{|\psi|^2}(t,\om)\tfrac{d\om}{(2\pi)^d}&\,.
\ea
$$
Then, by definition of $\cR_\psi(t)$, one has
$$
\ba
\cC[V,\cR_\psi(t),\cR_\psi(t)]A
=\int_{\bR^d}\widehat{V\!\star\!|\psi|^2}(t,\om)\la\psi(t,\cdot)|[E_\om,A]|\psi(t,\cdot)\ra_{\fH}\,I_{\fH_N}\tfrac{d\om}{(2\pi)^d}
\\
=\La\psi(t,\cdot)\bigg|\left[\int_{\bR^d}\widehat{V\!\star\!|\psi|^2}(t,\om)E_\om\tfrac{d\om}{(2\pi)^d},A\right]\bigg|\psi(t,\cdot)\Ra_{\fH}\,I_{\fH_N}
\\
=\la\psi(t,\cdot)|[V\!\star\!|\psi|^2(t,\cdot),A]|\psi(t,\cdot)\ra_{\fH}\,I_{\fH_N}&\,.
\ea
$$
In the special case where $A=|\phi\ra\la\phi|$, one has
$$
[V\!\star\!|\psi|^2(t,\cdot),A]=|V\!\star\!|\psi|^2(t,\cdot)\phi\ra\la\phi|-|\phi\ra\la V\!\star\!|\psi|^2(t,\cdot)\phi|\,,
$$
so that
$$
\cC[V,\cR_\psi(t),\cR_\psi(t)]A=-2i\Im\left(\overline{\la\phi|\psi(t,\cdot)\ra}\la\phi|V\!\star\!|\psi|^2(t,\cdot)|\psi(t,\cdot)\ra\right)I_{\fH_N}\,.
$$

On the other hand
$$
\cR_\psi(t)(|\phi\ra\la\phi|)=|\la\phi|\psi(t,\cdot)\ra|^2I_{\fH_N}\,,
$$
while
$$
\ba
\Ad(-\tfrac12\hb^2\Dlt_x)\cR_\psi(t)(|\phi\ra\la\phi|)=&-\cR_\psi(t)([-\tfrac12\hb^2\Dlt_x,|\phi\ra\la\phi|])
\\
=&-\la\psi(t,\cdot)|[-\tfrac12\hb^2\Dlt_x,|\phi\ra\la\phi|]|\psi(t,\cdot)\ra\,I_{\fH_N}
\\
=&2i\Im\left(\overline{\la\phi|\psi(t,\cdot)\ra}\la\phi|\tfrac12\hb^2\Dlt_x|\psi(t,\cdot)\ra\right)\,.
\ea
$$

Then, the fact that $\psi$ is a solution of the TDH equation implies that
$$
i\hb\d_t\la\phi|\psi(t,\cdot)\ra=\la\phi|-\tfrac12\hb^2\Dlt_x|\psi(t,\cdot)\ra+\la\phi|V\star|\psi(t,\cdot)|^2|\psi(t,\cdot)\ra\,,
$$
and therefore
$$
\ba
i\hb\d_t|\la\phi|\psi(t,\cdot)\ra|^2=&2i\Im\left(\overline{\la\phi|\psi(t,\cdot)\ra}\la\phi|-\tfrac12\hb^2\Dlt_x|\psi(t,\cdot)\ra\right)
\\
&+2i\Im\left(\overline{\la\phi|\psi(t,\cdot)\ra}\la\phi|V\!\star\!|\psi|^2(t,\cdot)|\psi(t,\cdot)\ra\right)\,,
\ea
$$
which obviously implies the desired result.
\end{proof}

\subsection{The Quantum Mean-Field Limit with Klimontovich Solutions}


As a first application of the notion of quantum Klimontovich solution discussed above, we present a derivation of the mean-field limit in quantum mechanics. This derivation is suboptimal, in particular because it assumes $\hat V\in L^1(\bR^d)$.
On the other hand the proof includes several features of interest for other applications.

\begin{Thm}
Let $V$ be a real-valued potential satisfying assumptions (H1)-(H4). Let $\psi\equiv\psi(t,x)$ be a solution of the Hartree equation
$$
\left\{
\ba
{}&i\hb\d_t\psi=-\tfrac12\hb^2\Dlt_x\psi+(V\star|\psi(t,\cdot)|^2)\psi\,,\qquad x\in\bR^d\,,
\\
&\psi\rstr_{t=0}=\psi^{in}\,,
\ea
\right.
$$
with initial data $\psi^{in}\in H^1(\bR^d)$ satisfying the normalization $\|\psi^{in}\|_{L^2(\bR^d)}=1$. For each $t\ge 0$, let 
$$
\Psi_N(t,\cdot):=e^{-it\cH_n/\hb}\Psi_N^{in}\qquad\text{ with }\Psi_N^{in}:=(\psi^{in})^{\otimes N}\,.
$$
Then the $N$-particle reduced density operator $R_{N:1}(t)$ associated to the wave function $\Psi_N(t,\cdot)$ satisfies
$$
\|\,R_{N:1}(t)-|\psi(t,\cdot)\ra\la\psi(t,\cdot)|\,\|\le\frac2{\sqrt{N}}\exp\left(\frac{2t\|\hat V\|_{L^1}}{(2\pi)^d\hb}\right)\,.
$$
\end{Thm}

\begin{proof} Let $t\mapsto A(t)\in\cL(\fH)$ be the solution of the linear von Neumann equation with time-dependent potential
$$
i\hb\d_tA(t)=[-\tfrac12\hb^2\Dlt_x+V_\psi(t,\cdot),A(t)]\,,\qquad A\rstr_{t=0}=A^{in}\,.
$$
Elementary computations left to the reader show that
$$
i\hb\d_t(\cM_N(t)-\cR_\psi(t))(A(t))=-\cC[V,\cM_N(t)-\cR_\psi(t),\cM_N(t)](A(t))\,.
$$
Let $S_N\in\cL(\fH_N)$ with $\|S_N\|\le 1$. Then
$$
\ba
\la\Psi_N^{in}|S_N(\cM_N(t)-\cR_\psi(t))(A(t))|\Psi_N^{in}\ra-\la\Psi_N^{in}|S_N(\cM_N^{in}-\cR_\psi(0))(A^{in})|\Psi_N^{in}\ra
\\
=\frac{i}{\hb}\int_0^t\int_{\bR^d}\hat V(\om)\la\Psi_N^{in}|S_N(\cM_N(s)-\cR_\psi(s))(E_\om^*)\cM_N(s)(E_\om A(s))|\Psi_N^{in}\ra\tfrac{d\om ds}{(2\pi)^d}
\\
-\frac{i}{\hb}\int_0^t\int_{\bR^d}\hat V(\om)\la\Psi_N^{in}|S_N\cM_N(s)(A(s)E_\om)(\cM_N(s)-\cR_\psi(s))(E_\om^*)|\Psi_N^{in}\ra\tfrac{d\om ds}{(2\pi)^d}
\\
=\!\frac{i}{\hb}\!\int_0^t\!\!\int_{\bR^d}\!\!\hat V(\om)\la\Psi_N^{in}|S_N\cM_N(s)([E_\om,A(s)])(\cM_N(s)-\cR_\psi(s))(E_\om^*)|\Psi_N^{in}\ra\tfrac{d\om ds}{(2\pi)^d}
\\
-\frac{i}{\hb}\int_0^t\int_{\bR^d}\hat V(\om)\la\Psi_N^{in}|S_N[\cM_N(s)(E_\om A(s)),\cM_N(s)(E_\om^*)]|\Psi_N^{in}\ra\tfrac{d\om ds}{(2\pi)^d}&\,.
\ea
$$
Then
$$
\ba
{}[\cM_N(s)(E_\om A(s)),\cM_N(s)(E_\om^*)]
\\
=e^{+is\cH_N/\hb}[\cM_N^{in}(E_\om A(s)),\cM_N^{in}(E_\om^*)]e^{-is\cH_N/\hb}
\\
=e^{+is\cH_N/\hb}\frac1{N^2}\sum_{k=1}^NJ_k[E_\om A(s),E_\om^*]e^{-is\cH_N/\hb}&\,,
\ea
$$
so that
$$
\|[\cM_N(s)(E_\om A(s)),\cM_N(s)(E_\om^*)]\|\le\frac2N\|A(s)\|\,.
$$
Set
$$
d_N(t):=\sup_{\|S_N\|\le 1}\sup_{\|B\|=1}|\la\Psi_N^{in}|S_N(\cM_N(t)-\cR_\psi(t))(B)|\Psi_N^{in}\ra|\,.
$$
Since
$$
\|\cM_N(s)([E_\om,A(s)])\|\le\|[E_\om,A(s)]\|\le 2\|A(s)\|\,,
$$
and since 
$$
\|\phi^{in}\|_\fH=1\implies\|A(s)\|=1\,,
$$
one has
$$
|\la\Psi_N^{in}|S_N\cM_N(s)([E_\om,A(s)])(\cM_N(s)-\cR_\psi(s))(E_\om^*)|\Psi_N^{in}\ra|\le 2d_N(s)\,,
$$
so that
$$
\ba
|\la\Psi_N^{in}|S_N(\cM_N(t)-\cR_\psi(t))(A(t))|\Psi_N^{in}\ra|
\\
\le|\la\Psi_N^{in}|S_N(\cM_N^{in}-\cR_\psi(0))(A^{in})|\Psi_N^{in}\ra|+\frac{2\|\hat V\|_{L^1}}{(2\pi)^d\hb}\int_0^t\left(d_N(s)+\frac1N\right)ds&\,.
\ea
$$
Since $A(t)$ runs through the unit ball of $\cL(\fH)$ as $A^{in}$ runs through the unit ball of $\cL(\fH)$, one finds that
$$
\ba
d_N(t)=&\sup_{\|S_N\|\le 1}\sup_{\|\phi^{in}\|_\fH\le 1}|\la\Psi_N^{in}|S_N(\cM_N(t)-\cR_\psi(t))(|\phi(t,\cdot)\ra\la\phi(t,\cdot|)|\Psi_N^{in}\ra|
\\
\le& d_N(0)+\frac{2\|\hat V\|_{L^1}}{(2\pi)^d\hb}\int_0^t\left(d_N(s)+\frac1N\right)ds\,,
\ea
$$
so that
$$
d_N(t)\le d_N(0)\exp\left(\frac{2t\|\hat V\|_{L^1}}{(2\pi)^d\hb}\right)+\frac1N\left(\exp\left(\frac{2t\|\hat V\|_{L^1}}{(2\pi)^d\hb}\right)-1\right)
$$
by Gronwall's inequality. 

With $S_N=I_{\fH_N}$, one finds that
$$
\ba
d_N(t)\ge&\sup_{\|\phi\|_\fH\le 1}|\la\Psi_N^{in}|\cM_N(t)(|\phi\ra\la\phi|)|\Psi_N^{in}\ra-\la\Psi_N^{in}|\cR_\psi(t)(|\phi\ra\la\phi|)|\Psi_N^{in}\ra|
\\
=&\sup_{\|\phi\|_\fH\le 1}|\la\phi|R_{N:1}(t)|\phi\ra-|\la\psi(t,\cdot)|\phi\ra|^2|=\|R_{N:1}(t)-|\psi(t,\cdot)\ra\la\psi(t,\cdot)|\|\,.
\ea
$$

On the other hand, by the Cauchy-Schwarz inequality, if $S_N\in\cL(\fH)$ satisfy $\|S_N\|\le 1$ while $\|B^{in}\|=1$, one has
$$
\ba
|\la\Psi_N^{in}|S_N(\cM_N^{in}-\cR_\psi(0))(B^{in})|\Psi_N^{in}\ra|
\\
\le\|S\Psi_N^{in}\|\|(\cM_N^{in}-\cR_\psi(0))(B^{in})\Psi_N^{in}\|
\\
\le\|(\cM_N^{in}-\cR_\psi(0))(B^{in})\Psi_N^{in}\|&\,.
\ea
$$
Assuming that
$$
\Psi_N^{in}=\left(\psi^{in}\right)^{\otimes N}\,,
$$
one has
$$
\ba
\|(\cM_N^{in}-\cR_\psi(0))(B^{in})\Psi_N^{in}\|^2
\\
=\la\Psi_N^{in}|(\cM_N^{in}-\cR_\psi(0))(B^{in})^*(\cM_N^{in}-\cR_\psi(0))(B^{in})|\Psi_N^{in}\ra
\\
=\!\tfrac1N\la\psi^{in}|(B^{in})^*B^{in}|\psi^{in}\ra\!+\!\tfrac{N-1}N|\la\psi^{in}|B^{in}|\psi^{in}\ra|^2\!-\!2|\la\psi^{in}|B^{in}|\psi^{in}\ra|^2\!+\!|\la\psi^{in}|B^{in}|\psi^{in}\ra|^2
\\
=\!\tfrac1N\left(\la\psi^{in}|(B^{in})^*B^{in}|\psi^{in}\ra\!-\!|\la\psi^{in}|B^{in}|\psi^{in}\ra|^2\right)&\,.
\ea
$$
Therefore
$$
\|(\cM_N^{in}-\cR_\psi(0))(B^{in})\Psi_N^{in}\|^2\le\tfrac1N\|B^{in}\psi^{in}\|_\fH^2\le\tfrac1N\,,
$$
so that $d_N(0)\le\tfrac1{\sqrt{N}}$.
Gathering together all these inequalities shows that
$$
\|R_{N:1}(t)-|\psi(t,\cdot)\ra\la\psi(t,\cdot)|\|\le\frac1{\sqrt{N}}\exp\left(\frac{2t\|\hat V\|_{L^1}}{(2\pi)^d\hb}\right)+\frac1N\left(\exp\left(\frac{2t\|\hat V\|_{L^1}}{(2\pi)^d\hb}\right)-1\right)\,,
$$
ultimately leading to the desired inequality.
\end{proof}

\subsection{The Quantum Mean-Field Limit: Coulomb Interaction}


In the context of atomic physics, one often has to consider charged particles interacting through a repulsive Coulomb potential. In that case, the mean-field limit cannot be proved by the simple argument presented in the previous section.
However, the notion of quantum Klimontovich solution can be used also in this case, in the following manner.

Let $V$ satisfy (H1)-(H3) and
$$
V^2\le C(I-\Dlt)\leqno{(\mathrm H5)}
$$
for some constant $C>0$, in the sense of operators on $\fH$. (In space dimension $d=3$, the Hardy inequality, which can be put in the form\footnote{To see that $4$ is optimal, minimize in $\a>0$ the expression
$$
\int_{\bR^3}\left|\grad u+\a\frac{x}{|x|^2}u\right|^2dx\,.
$$}
$$
\frac1{|x|^2}\le 4(-\Dlt)
$$
implies that the Coulomb potential satisfies (H5).)

Let $\psi^{in}\in H^1(\bR^d)$ satisfy $\|\psi^{in}\|_{L^2}=1$, and let $\psi$ be the solution of the Hartree equation
$$
i\hb\d_t\psi(t,x)=-\tfrac12\hb^2\Dlt_x\psi(t,x)+(V\star|\psi(t,\cdot)|^2)(x)\psi(t,x)\,,\quad\psi\rstr_{t=0}=\psi^{in}\,.
$$

\begin{Thm}\label{T-Coulh=1}\cite{IBPFG}
Under the assumptions above, let $\cM_N(t)$ be the $N$-particle Klimontovich solution associated to the quantum Hamiltonian
$$
\cH_N=\sum_{k=1}^N-\tfrac12\hb^2\Dlt_{x_k}+\tfrac1N\sum_{1\le k<l\le N}V(x_k-x_l)\,.
$$
Then

\noindent
(1) one has
$$
\ba
i\hb\d_t\cM_N(t)(I_\fH-|\psi(t,\cdot)\ra\la\psi(t,\cdot)|)
\\
=\cC[V,\cM_N(t)-\cR_\psi(t),\cM_N(t)](I_\fH-|\psi(t,\cdot)\ra\la\psi(t,\cdot)|)&\,;
\ea
$$
(2) the interaction operator $\cC(V,\cM_N(t)-\cR(t),\cM_N(t))(I-|\psi(t,\cdot)\ra\la\psi(t,\cdot)|)$ is skew-adjoint on $\fH_N$ and satisfies the operator inequality
$$
\ba
\pm i\cC(V,\cM_N(t)-\cR(t),\cM_N(t))(I-|\psi(t,\cdot)\ra\la\psi(t,\cdot)|)
\\
\le 6L(t)\left(\cM_N(t)(I-|\psi(t,\cdot)\ra\la\psi(t,\cdot)|)+\tfrac{2}NI_{\fH_N}\right)&\,,
\ea
$$
where
$$
L(t):=\sqrt{C}\|\psi(t,\cdot)\|_{H^1}\,;
$$
(3) the $m$-particle reduced density operator $R_{N:m}(t)$ associated with the $N$-particle wave function
$$
\Psi_N(t,\cdot):=e^{-it\cH_n/\hb}\Psi_N^{in}\,,\quad\text{ with }\Psi_N^{in}:=(\psi^{in})^{\otimes N}
$$
satisfies
$$
\|R_{N:m}(t)-|\psi(t,\cdot)\ra\la\psi(t,\cdot)|^{\otimes m}\|\le 4\sqrt{\frac{m}N}\exp\left(\tfrac{3}\hb\int_0^tL(s)ds\right)
$$
for each $t\ge 0$ and each $m=1,\ldots,N$.
\end{Thm}

This is a reformulation of an earlier result by Pickl \cite{Pickl09} and Knowles-Pickl \cite{KnowPickl} in terms of the quantum Klimontovich solution. Pickl's original idea \cite{Pickl09} was to consider the quantity
$$
\cE(t):=1-\la\psi(t,\cdot)|R_{N:1}(t)|\psi(t,\cdot)\ra\,,
$$
and to prove that
$$
\tfrac{d}{dt}\cE(t)\le 10\|V\|_{L^{2r}}\|\psi(t,\cdot)\|_{L^{2r'}}\left(\cE(t)+\tfrac1N\right)\,.
$$
by a clever decomposition of $\tfrac{d}{dt}\cE(t)$ into the sum of three terms to be analyzed separately. Observing instead that
$$
\cE(t)=\la\Psi_N^{in}|\cM_N(t)(I_\fH-|\psi(t,\cdot)\ra\la\psi(t,\cdot)|)|\Psi_N^{in}\ra
$$
suggests the idea of using the equation satisfied by the quantum Klimontovich solution $\cM_N(t)$ to write an \textit{operator} inequality, specifically statement (2) in the theorem above, instead of the \textit{scalar} inequality satisfied by 
$\tfrac{d}{dt}\cE(t)$ as in Pickl's original work \cite{Pickl09}. One recovers the Knowles-Pickl, or the Pickl estimate by evaluating the operators in (2) on the quantum state defined by the wave function $\Psi_N^{in}$.

One essential difference between Theorem \ref{T-Coulh=1} and Theorem \ref{T-MFh=1} is in the use of the single particle test wave function. In the proof of Theorem \ref{T-MFh=1}, one considers the operator
$$
(\cM_N(t)-\cR_\psi(t))(|\phi(t,\cdot)\ra\la\phi(t,\cdot)|)
$$
where $\phi(t,\cdot)$ is any wave function propagated by the mean-field dynamics defined by the Hartree solution $\psi$, whereas in Theorem \ref{T-Coulh=1} one considers the operator
$$
\cM_N(t)(I_\fH-|\psi(t,\cdot)\ra\la\psi(t,\cdot)|)
$$
where $\psi$ is the target Hartree solution. Because of the specifics of the latter choice, the approach described in Theorem \ref{T-Coulh=1} applies only to \textit{pure} quantum states, in other words, on quantum states which can be 
described by means of a (single) wave function, and not to \textit{mixed} states, i.e. quantum states described by means of a density operator --- see lecture 3 below for a brief description of this notion. On the contrary, the proof of 
Theorem \ref{T-MFh=1} can be easily generalized to mixed states.

\subsection{Miscellaneous Remarks}


In this lecture, we have chosen to describe the quantum mean-field limit in terms of quantum Klimontovich solutions because of the novelty of this approach, and also because it parallels the classical theory presented in lecture 1.
However, this is by no means the only way in which the quantum mean-field limit can be justified rigorously.

\smallskip
\noindent
(1) Historically, the first rigorous justification of the quantum mean-field limit was obtained by analyzing the \textit{BBGKY hierarchy}. Starting from the $N$-particle Schr\"odinger equation
$$
i\hb\d_t\Psi_N(t,X_N)=\cH_N\Psi_N(t,X_N)\,,
$$
one easily obtains a differential equation for the single-particle reduced density operator $R_{N:1}(t)$. Because of the $2$-body interaction potential $V$ in the quantum Hamiltonian $\cH_N$, the differential equation for $R_{N:1}$ involves
the $2$-particle reduced density operator $R_{N:2}(t)$. Therefore, one writes a differential equation for the operator $R_{N:2}(t)$, but this equation involves the $3$-particle reduced density operator $R_{N:3}(t)$. More generally, for each 
integer $k<N$, the differential equation satisfied by the $k$-particle reduced density operator $R_{N:k}(t)$ involves the $k+1$-particle reduced density operator $R_{N:k+1}(t)$. One obtains in this way a sequence of differential equations
for $R_{N:k}(t)$ for all $k\ge 1$ --- with the convention that
$$
R_{N:N}(t):=|\Psi_N(t,\cdot)\ra\la\Psi_N(t,\cdot)|\,,\qquad R_{N:k}(t)=0\text{ if }k>N\,.
$$
This sequence of differential equations is known as the ``BBGKY hierarchy'' (named after Bogolyubov, Born, Green, Kirkwood and Yvon). The idea is to pass to the limit in each equation of this hierarchy, i.e. for each $k\ge 1$ in the limit
as $N\to\infty$, and to prove by some uniqueness argument akin to the Cauchy-Kovalevska theorem that
$$
R_{N:k}(t)\to|\psi(t,\cdot)\ra\la\psi(t,\cdot)|
$$
in some appropriate topology, where $\psi$ is the Hartree solution. The first proof along this line is due to Spohn and sketched in \cite{Spohn80}; more details can be found in \cite{BGM}, and the interpretation in terms of the Cauchy-Kovalevska 
theorem is presented in \cite{BEGMY}. Incidentally, it is interesting to notice that the BBGKY approach was not used on the classical mean-field limit, at least until very recently: see \cite{Duerinckx}.

\smallskip
\noindent
(2) Spohn's derivation of the quantum mean-field limit by means of the BBGKY hierarchy relies on the assumption that the interaction potential is $V$ even (i.e. satisfies (H1)) and that
$$
V\in L^\infty(\bR^d)\,.
$$
Therefore, this derivation did not include the physically interesting case of a repulsive Coulomb interaction between identical charged particles. This case was handled later by Erd\"os and Yau \cite{EY} (see also \cite{BEGMY}). While Spohn's 
original argument involved estimates in \textit{trace-norm} (see lecture 3 for a definition of the trace of an nonnegative operator on $\fH$) for the reduced density operators $R_{N:k}$, one of the key ideas in \cite{EY} was to use \textit{weighted}
trace-norms on $R_{N:k}(t)$ involving cross-derivatives in the $k$-tuples of position variables.

\smallskip
\noindent
(3) One of the shortcomings of the BBGKY approach is the lack of quantitative information on the convergence rate obtained by this method. As explained above, this method involves a uniqueness argument \`a la Cauchy-Kovalevska, which is
therefore very far from a stability estimate. 

For that reason, Rodnianski and Schlein \cite{RodSch} proposed a convergence rate estimate for the mean-field limit based on a formulation of the problem in a 2nd quantization setting, in other words in the bosonic Fock space. They obtained 
a $O(1/\sqrt{N})$ convergence rate, consistent with the estimate obtained in Theorem \ref{T-Coulh=1}. We shall not give too many details on this approach, which requires being acquainted with the fundamental notions of 2nd quantization
(Fock space, creation/annihilation operators, number operator\dots) Section 2 of \cite{RodSch} provides a very clear introduction to this material. Suffices it to say that the $O(1/\sqrt{N})$ convergence rate in \cite{RodSch} is obtained under the 
assumption that $V$ satisfies (H1)-(H5) (exactly as in Theorem \ref{T-Coulh=1}).

\smallskip
\noindent
(4) That the same assumptions (H1) and (H5) on the interaction potential appear in the 2nd quantization approach \cite{RodSch}, in the Knowles-Pickl result \cite{KnowPickl} and in the quantum Klimontovich solutions approach \cite{IBPFG} is 
hardly surprising, for the following reasons. We have already explained in the paragraph following Theorem \ref{T-Coulh=1} between Pickl's approach \cite{Pickl09} and the the quantum Klimontovich solutions approach \cite{IBPFG}. 

Using the 2nd quantization approach as in \cite{RodSch} to prove the quantum mean-field limit seems unnecessarily complicated, since the quantum dynamics corresponding to the quantum Hamiltonian $\cH_N$ preserves the particle number $N$,
whereas the formalism of Fock spaces is specifically designed to handle situations where the particle number varies (for instance due to disintegration). Since the quantum dynamics $e^{-it\cH_N/\hb}$ leaves the particle number $N$ invariant, 
the $N$-particle sector in the Fock space is invariant under the dynamics considered in \cite{RodSch}. The restriction of this dynamics to the $N$-particle sector corresponds precisely to the equation for the quantum Klimontovich solution obtained in
Theorem \ref{T-qKlim}. More precisely, using freely the notation in \cite{RodSch}, one has
$$
a^*(\phi)a(\phi)=0\oplus\bigoplus_{N\ge 1}N\cM_N^{in}(|\phi\ra\la\phi|)
$$
for each $\phi\in\fH$ such that $\|\phi\|_\fH=1$. Consider the Hamiltonian in Fock space defined by the formula
$$
\cH:=\tfrac1{2m}\hb^2\int_{\bR^d}dx\grad_xa_x^*\cdot\grad_xa_x+\iint_{\bR^{2d}}dxdyV(x-y)a^*_xa^*_ya_ya_x\,,
$$
where $m$ is the mass of one particle. The (unbounded) operator $\cH$ defines a unitary group $e^{-it\cH/\hb}$ in Fock space leaving the $N$-particle sector invariant for each $N\ge 0$. Up to some appropriate rescaling of the time variable,
and setting $m=1/N$, the restriction to the $N$-particle sector of 
$$
e^{-it\cH/\hb}a^*(\phi)a(\phi)e^{+it\cH/\hb}
$$
is expected to coincide with
$$
N\cM_N(Nt)(|\phi\ra\la\phi|)\,.
$$
These remarks will be presented in detail in \cite{FG2q}; they provide the missing link between the second quantization approach and the quantum Klimontovich solution approach for the quantum mean-field limit.

\smallskip
\noindent
(5) So far we have considered Hamiltonians of the form
$$
\cH_N=\sum_{k=1}^N-\tfrac1{2m}\hb^2\Dlt_{x_k}+\sum_{1\le k<l\le N}V(x_k-x_l)\,,
$$
in which the potential energy comes only from the binary interaction between the particles. All the mathematical tools presented in this lecture apply to more general Hamiltonians of the form
$$
\cH_N=\sum_{k=1}^N(-\tfrac1{2m}\hb^2\Dlt_{x_k}+U(x_k))+\sum_{1\le k<l\le N}V(x_k-x_l)\,,
$$
where $U$ is an external potential acting separately on each particle. (In the case of atomic physics, one could think of $V(x_k-x_l)$ as the Coulomb repulsive interaction between electrons at the positions $x_k$ and $x_l$, whereas $U(x_k)$
would be the attracting potential exerted by the nuclei on an electron at the position $x_k$.) We shall not dwell on this matter any longer, and leave it to the reader to modify all the statements in the present lecture in order to handle this more
general class of quantum Hamiltonians.

\smallskip
\noindent
(6) There is also the problem of deriving a theory for fluctuations around the mean-field limit, both for the classical and the quantum dynamics. This problem has been studied in \cite{BraunHepp} in the classical case (for regular potentials).
More recently, the quantum analogue of this problem has been treated in \cite{GBAKirkSchlein} in the 2nd quantization setting, and under the same assumptions as in \cite{RodSch}. It should be possible to express this result in terms of
quantum Klimontovich solutions (to avoid the unnatural appearance of Fock'space in a problem where the particle number is constant): see \cite{FGSS}.

\smallskip
\noindent
(7) In this lecture, we have discussed only the case of $N$ bosons. However, in atomic physics, electrons, which are fermions, are the interacting particles of interest. Because of the Pauli exclusion principle, the kinetic energy of $N$ identical
fermions in a box of unit volume in $\bR^3$ grows at least as $N^{5/3}$. In order for the kinetic energy and the potential energy in 
$$
\sum_{k=1}^N-\tfrac1{2m}\hb^2\Dlt_{x_k}+\sum_{1\le k<l\le N}V(x_k-x_l)
$$
to be of the same order of magnitude, the coupling constant in front of the potential energy should be of order $1/N^{1/3}$, instead of $1/N$ as in the bosonic case. For this reason, we scale the time as $t=N^{-1/3}\tau$, which leads us to consider 
the scaled Schr\"odinger equation
$$
i\hb N^{1/3}\d_\tau\Psi_N=\sum_{k=1}^N-\tfrac1{2}\hb^2\Dlt_{x_k}\Psi_N+\frac1{N^{1/3}}\sum_{1\le k<l\le N}V(x_k-x_l)\Psi_N\,.
$$
Set $\eps:=N^{-1/3}$; multiplying both sides of the equation above by $\eps^2$ shows that
$$
i\hb\eps\d_\tau\Psi_N=\sum_{k=1}^N-\tfrac1{2}\hb^2\eps^2\Dlt_{x_k}\Psi_N+\frac1N\sum_{1\le k<l\le N}V(x_k-x_l)\Psi_N\,.
$$
In other words, the mean-field limit for fermions corresponds to studying the equation above with $N\to\infty$ and $\eps=N^{-1/3}\to 0$. Letting $\eps\to 0$ in the Schr\"odinger equation above with $\hb=1$ corresponds to the classical limit of 
quantum mechanics. Thus, the mean-field limit for fermions must involve mathematical techniques combining both the classical limit and the mean-field limit for the Schr\"odinger equation, in the distinguished asymptotics $\eps^3N=1$. This topic 
will be studied in more details in lecture 3. 

Alternatively, one could keep the same scaling as for the bosonic mean-field limit in the case of fermions, but this will lead to situations where the potential energy is negligible compared to the kinetic energy of the $N$-fermion system. In this setting
however, one can check that the time-dependent Hartree-Fock equation naturally appears in that limit --- but of course, one should instead think of the mean-field equation obtained in this way as an asymptotic correction to the (uninteresting) free 
dynamics, so that the accuracy of the approximation becomes of interest in this case. See \cite{BGGM1,BGGM2}. 

Otherwise, the mean-field limit for $N$ fermions in the scaling for which the kinetic and the potential energies are comparable, leading to the time-dependent Hartree, or the time-dependent Hartree-Fock equations, has been studied in \cite{BPS}
and in \cite{BRSS,BJPSS}. (Notice that these references use the formalism of 2nd quantization, as in (4) above, but in the fermionic setting.)


\section{Lecture 3: Mean-Field and Classical Limits \\ in Quantum Mechanics}


Let us begin with the following diagram in order to explain what has been achieved so far.

\smallskip

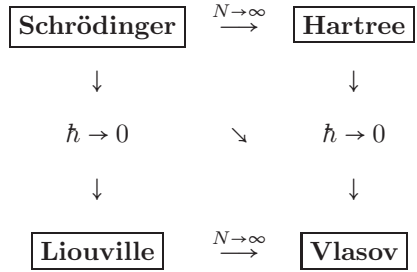
\begin{center}
\begin{tabular}{ccc}
{\fbox{\bf Schr\"odinger}}& {$\stackrel{N\to\infty}{\longrightarrow}$}& {\fbox{\bf Hartree}} 
\\ [3mm]
{$\downarrow$} & & {$\downarrow$} 
\\[3mm]
{${\hbar\to0}$}&$\searrow$& {${\hbar\to 0}$} 
\\ [3mm]
$\downarrow$&& $\downarrow$ 
\\ [3mm]
{\fbox{\bf Liouville}}& {$\stackrel{N\to\infty}{\longrightarrow}$}&{\fbox {\bf Vlasov}} 
\end{tabular}
\end{center}

\smallskip
The lower horizontal arrow corresponds to the limit studied in lecture 1, with convergence rate expressed in terms of the Dobrushin inequality involving the Monge-Kantorovich distance.

The upper horizontal arrow corresponds to the limit studied in lecture 2, with convergence rate given by Theorem \ref{T-MFh=1} --- or by Theorem \ref{T-Coulh=1} in the Coulomb case.

The first results on the joint mean-field and classical limit, i.e. the oblique arrow in the diagram, without any distinguished limit scaling, are \cite{GraffiMartiPulvi,PezzoPulvi}.

\subsection{Dynamics of $N$-Body Density Operators}


First we recall the notion of density operator in quantum mechanics. This is the quantum analogue of the notion of distribution function in kinetic theory.

\subsubsection{Quantum Density Operators}


Let $0\le T=T^*\in\cL(\fH)$ where $\fH$ is a separable Hilbert space, and $(e_n)_{n\ge 1}$ a complete orthonormal system in $\fH$. The trace of $T$ is defined by the formula
$$
\Tr_\fH(T):=\sum_{n\ge 1}\la e_n|T|e_n\ra\in[0,+\infty]\,.
$$
One easily checks that the right hand side of this formula is independent of the choice of the complete orthonormal system $(e_n)_{n\ge 1}$ of $\fH$. (The notion of trace of a nonnegative operator is analogous to the integral of a nonnegative 
measurable function: it always exists as an element of $[0,+\infty]$.)

\smallskip
The set of \textit{density operators} on $\fH$ is
$$
\cD(\fH):=\{R\in\cL(\fH)\text{ s.t. }R=R^*\ge 0\text{ and }\Tr_\fH(R)=1\}\,.
$$
When $\fH=L^2(\bR^d)$, one should think of the set of density operators $\cD(\fH)$ as the quantum analogue of the set $\cP(\bR^d\times\bR^d)$ of Borel probability measures on phase space.

\smallskip
\noindent
\textbf{Example.} If $(\psi_n)_{n\ge 1}$ is an orthonormal system of wave functions, not necessarily complete,
$$
R=\sum_{n\ge 1}\l_n|\psi_n\ra\la\psi_n|\in\cD(\fH)\iff\l_n\ge 0\quad\text{ and }\sum_{n\ge 1}\l_n=1\,.
$$

\smallskip
The quantum analogue of  $\cP(\bR^d\times\bR^d)$ (the set of Borel probability measures on phase space with finite second order moments) is the set of \textit{finite energy} density operators:
$$
\cD_2(\fH):=\{R\in\cD(\fH)\text{ s.t. }\Tr_\fH(R^{1/2}(|x|^2-\Dlt_x)R^{1/2})<\infty\}\,.
$$
(In this terminology, finite energy refers to the quantum harmonic oscillator 
$$
-\tfrac1{2m}\hb^2\Dlt_x+\tfrac12m\om^2|x|^2\,,
$$
where $m$ is the particle mass and $\om$ the oscillation frequency.)

\smallskip
In the case of systems of $N$ indistinguishable particles moving in $\bR^d$, the relevant density operators are \textit{symmetric $N$-particle density operators} on the $N$-particle Hilbert space $\fH_N=\fH^{\otimes N}\simeq L^2(\bR^{dN})$
(if $\fH=L^2(\bR^d)$). 

The set of symmetric $N$-particle density operators is
$$
\cD^s(\fH_N):=\{R_N\in\cD(\fH_N)\text{ s.t. }U_\si R_NU_\si^*=R_N\text{ for all }\si\in\fS_N\}\,,
$$
where $U_\si$ is the representation of the symmetric group $\fS_N$ in $\fH_N$, defined by the formula
$$
U_\si\Psi_N(X_N):=\Psi_N(x_{\si^{-1}(1)},\ldots,x_{\si^{-1}(N)})\,,
$$
for all $\Psi_N\in\fH_N$.

\smallskip
For each symmetric, $N$-particle density operator $R_N\in\cD_s(\fH_N)$, one defines its $k$-particle marginal $R_{N:k}\in\cD^s(\fH_k)$ as follows. If $r_N(X_N,Y_N)$ is an integral kernel\footnote{If $R\in\cD(\fH)$, then $R^{1/2}$ is a 
Hilbert-Schmidt operator on $\fH$, and has therefore an integral kernel $r_{1/2}\equiv r_{1/2}(x,y)\in L^2(\bR^d_x\times\bR^d_y)$. Since $R^{1/2}$ is self-adjoint, $R$ has integral kernel
$$
r(x,y):=\int_{\bR^d}r_{1/2}(x,z)\overline{r_{1/2}(y,z)}dz\,.
$$
Of course, $r$ can be modified on a Lebesgue negligible set, and this is why we speak of ``an integral kernel''. Notice however that the integral kernel $r$ defined by the formula above has the following remarkable property: by the Cauchy-Schwarz
inequality
$$
\int_{\bR^d}|r(x+h,x)-r(x,x)|dx\le\|r_{1/2}\|_{L^2(\bR^{2d})}\left(\iint_{\bR^{2d}}|r_{1/2}(x+h,z)-r_{1/2}(x,z)|^2dxdz\right)^{1/2}\to 0
$$
as $|h|\to 0$, by the continuity of the action of $\bR^{2d}$ by translation on $L^2(\bR^{2d})$. In other words, the integral kernel $r$ above is such that $h\mapsto r(x+h,x)$ belongs to $C(\bR^d;L^1(\bR^d))$.This is a special case of Lemma 2.1 (1)
in \cite{BGM}. In particular, one has
$$
\Tr_\fH(R)=\int_{\bR^d}r(x,x)dx\,,
$$
and the observation above justifies the existence of the integral in the right-hand side of this identity.} of $R_N$, its $k$-th marginal $R_{N:k}$ has integral kernel
$$
r_k(X_k,Y_k)=\int_{\bR^{d(N-k)}}r_N(X_k,Z_{k,N},Y_k,Z_{k,N})dZ_{k,N}\,,
$$
where we recall that
$$
Z_{k,N}:=(z_{k+1},\ldots,z_N)\,.
$$
\textbf{Example.} If $R_N=|\Psi_N\ra\la\Psi_N|$ with $\Psi_N\in\fH_N$ symmetric, then $R_N\in\cD^s(\fH_N)$ and $R_{N:1}$ is the first reduced density operator defined in Lecture 2.

\smallskip
Then we introduce the quantum dynamics of $N$-body density operators. First, we recall the $N$-particle quantum Hamiltonian
$$
\sum_{k=1}^N-\tfrac1{2m}\hb^2\Dlt_{x_k}+\sum_{1\le k<l\le N}V(x_k-x_l)\,.
$$
Pick a length scale $\ell>0$ and an energy scale $W$, and define dimensionless position variables and interaction potential by the formulas
$$
\hat x_k=\frac{x_k}\ell\,,\qquad \text{ and }\quad\hat V(\hat x_k-\hat x_l)=\frac{V(x_k-x_l)}{W}\,.
$$
Then
$$
\frac1{NW}\left(\sum_{k=1}^N-\tfrac1{2m}\hb^2\Dlt_{x_k}+\sum_{1\le k<l\le N}V(x_k-x_l)\right)=\sum_{k=1}^N-\tfrac12\eps^2\Dlt_{\hat x_k}+\tfrac1N\sum_{1\le k<l\le N}\hat V(\hat x_k-\hat x_l)\,,
$$
where $\eps$ is the dimensionless parameter defined by
$$
\eps^2:=\frac{\hb^2}{Nm\ell^2W}\ll 1\,.
$$

Henceforth, dropping all hats on the scaled variables, we arrive at the dimensionless Hamiltonian
\be\lb{HNsemicl}
\cH_N=\sum_{k=1}^N-\tfrac12\eps^2\Dlt_{x_k}+\tfrac1N\sum_{1\le k<l\le N}V(x_k-x_l)\,,
\ee
where
$$
N\gg 1\,,\quad\text{ and }\quad\eps\ll 1\,.
$$
We have seen in lecture 2 assumptions on $V$ such that the differential operator $\cH_N$ above has a self-adjoint extension (obviously unbounded) on $\fH_N$, still denoted $\cH_N$. In particular, by Stone's theorem, $e^{-\frac{it\cH_N}{\eps}}$ 
is a unitary group on $\fH_N$. Starting from $R_N^{in}\in\cD^s(\fH_N)$, we define
$$
R_N(t)=e^{-\frac{it\cH_N}{\eps}}R_N^{in}e^{+\frac{it\cH_N}{\eps}}\in\cD^s(\fH_N)\,.
$$
\textbf{Example.} For instance, if $R^{in}=|\Psi^{in}_N\ra\la\Psi^{in}_N|$ is the pure state associated with the $N$-particle wave function $\Psi^{in}_N$, then
$$
R_N(t)=e^{-\frac{it\cH_N}{\eps}}(|\Psi^{in}_N\ra\la\Psi^{in}_N|)e^{+\frac{it\cH_N}{\eps}}=\bigg|e^{-\frac{it\cH_N}{\eps}}\Psi^{in}_N\Ra\La e^{-\frac{it\cH_N}{\eps}}\Psi^{in}_N\bigg|\,.
$$

\subsection{Quantum-to-Classical Wasserstein Pseudo-Distance}


We have seen in lecture 1 how the mean field limit in classical mechanics could be couched in terms of the Monge-Kantorovich, or Wasserstein distance of exponent $1$. 

In order to arrive at an analogous quantitative estimate for the joint mean-field and classical limits represented by the diagonal arrow in the diagram at the beginning of this lecture, we first construct an analogue of this metric. Of course
the conceptual difficulty is that one seeks to compare apparently unrelated objects, namely a (classical) probability density on phase space $\bR^d\times\bR^d$, and a (quantum) density operator on $\fH=L^2(\bR^d)$.

\subsubsection{Coupling Quantum and Classical Densities}


As always in the definition of Monge-Kantorovich, or Wasserstein distances, the first task is to define a notion of \textit{coupling} of $R\in\cD(\fH)$ and $f\in\cP(\bR^d\times\bR^d)$. Such a coupling will be an operator-valued map
$$
(x,\xi)\mapsto Q(x,\xi)=Q(x,\xi)^*\in\cL(\fH)\,,
$$
such that
$$
Q(x,\xi)\ge 0\,,\qquad\Tr_\fH(Q(x,\xi))=f(x,\xi)\,,\qquad\int_{\bR^{2d}}Q(x,\xi)dxd\xi=R\,.
$$
The set of all couplings of the probability density $f$ and of the density operator $R$ is denoted $\cC(f,R)$.

\smallskip
\noindent
\textbf{Example.} The map
$$
f\otimes F:\,(x,\xi)\mapsto f(x,\xi)R\qquad\hbox{ belongs to }\cC(f,R)\,.
$$
In particular, the set $\cC(f,R)$ is never empty.

\subsubsection{Pseudo-Distance Between Quantum and Classical Densities}


Once a notion of coupling between a classical and a quantum density has been defined, the next task to fulfil in order to arrive at an analogue of the Monge-Kantorovich, or Wasserstein metric is to propose a notion of \textit{cost} for
transporting matter from the phase space point at the position $x$ with momentum $\xi$ to the ``quantum point'' corresponding to the position $y$ and to the (rescaled) momentum $-i\eps\grad_y$. Of course, there are no ``quantum 
points'', but if one has in mind the square Euclidean distance between phase space points, this immediately suggests the cost
$$
c_\eps(x,\xi):=|x-y|^2+|\xi+i\eps\grad_y|^2\,.
$$
At variance with the (classical) transport cost between two phase space points, this new object is an (unbounded) operator-valued function of the classical phase space coordinates $(x,\xi)$. More precisely, it is a harmonic oscillator
in the quantum position variable $y$, shifted in phase space by $(x,\xi)$.

\begin{Def}\cite{FGTPArma} For all $f\in L^1\cap\cP_2(\bR^d\times\bR^d)$ and all $R\in\cD_2(\fH)$, the quantum-to-classical Wasserstein pseudo-distance between $R$ and $f$ is defined by the formula
$$
E_\eps(f,R):=\inf_{Q\in\cC(f,R)}\sqrt{\int_{\bR^{2d}}\Tr_\fH(Q(x,\xi)^{\frac12}c_\eps(x,\xi)Q(x,\xi)^{\frac12})dxd\xi}\,.
$$
\end{Def}

Notice the different normalizations of the transport cost $c_\eps$ in \cite{FGTPArma} and in the present paper.
 
The quantity in the right-hand side of the formula above is always finite, as can beween from the following elementary argument. First
$$
c_\eps(x,\xi)\le 2(|x|^2+|\xi|^2)I_\fH+2c_\eps(0,0)
$$
(this is an inequality between unbounded self-adjoint operators on $\fH$ parametrized by the phase space point $(x,\xi)$). Then, for each $Q\in\cC(f,R)$, one has
$$
\ba
\int_{\bR^{2d}}\Tr_\fH(Q(x,\xi)^{\frac12}c_\eps(x,\xi)Q(x,\xi)^{\frac12})dxd\xi
\\
\le 2\int_{\bR^{2d}}\left((|x|^2+|\xi|^2)\Tr_\fH(Q(x,\xi))+\Tr_\fH(Q(x,\xi)^{\frac12}c_\eps(0,0)Q(x,\xi)^{\frac12})\right)dxd\xi
\\
=2\int_{\bR^{2d}}(|x|^2+|\xi|^2)f(x,\xi)dxd\xi+2\Tr_\fH\left(c_\eps(0,0)^{\frac12}\int_{\bR^{2d}}Q(x,\xi)dxd\xi\,c_\eps(0,0)^{\frac12}\right)
\\
=2\int_{\bR^{2d}}(|x|^2+|\xi|^2)f(x,\xi)dxd\xi+2\Tr_\fH\left(c_\eps(0,0)^{\frac12}Rc_\eps(0,0)^{\frac12}\right)
\\
=2\int_{\bR^{2d}}(|x|^2+|\xi|^2)f(x,\xi)dxd\xi+2\Tr_\fH\left(R^{\frac12}c_\eps(0,0)R^{\frac12}\right)<\infty&\,,
\ea
$$
since $f\in\cP_2(\bR^d\times\bR^d)$ and all $R\in\cD_2(\fH)$. Finally, the inf in the definition of $E_\eps$ is finite since $\cC(f,R)$ is nonempty.

\smallskip
One can prove that quantum-to-classical Monge-Kantorovich or Wasserstein pseudo-distance satisfies the following triangle inequality.

\begin{Thm} [Triangle inequality] 
For all $f,g\in L^1\cap\cP_2(\bR^d\times\bR^d)$ and all density operator $R\in\cD_2(\fH)$, one has
$$
E_\eps(f,R)\le\MKd(f,g)+E_\eps(g,R)\,.
$$
\end{Thm}

See Theorem 3.5 in \cite{FGTPJmpa} for a proof of this important result. In particular, this implies that the function $f\mapsto E_\eps(f,R)$ is a nonexpanding map from $L^1\cap\cP_2(\bR^d\times\bR^d)$ equipped with the Monge-Kantorovich 
or Wasserstein metric of exponent $2$ to the real line with its usual metric defined by the absolute value:
$$
|E_\eps(f,R)-E_\eps(g,R)|\le\MKd(f,g)\,.
$$
This property can be used to extend the definition of $E_\eps(f,R)$ by a density argument to the case where $f\in\cP_2(\bR^d\times\bR^d)$ is a probability \textit{measure} and not a probability \textit{density} (with respect to the phase space 
Lebesgue measure).

\smallskip
Now, the pseudo-metric $E_\eps$ remains somewhat mysterious, and it would be helpful to have examples for which $E_\eps$ can be computed explicitly. 

For instance, the Monge-Kantorovich, or Wasserstein distance between two Dirac measures in phase space is easily computed (see Remark 7.5 (ii) 
in \cite{VillaniAMS}):
$$
\MKp(\de_{x_1,\xi_1},\de_{x_2,\xi_2})=\sqrt{|x_1-x_2|^2+|\xi_1-\xi_2|^2}\,.
$$
(In other words, the Monge-Kantorovich, or Wasserstein distance between two Dirac measures is equal to the Euclidean distance between the phase space points where the two Dirac measures are concentrated.) Another important example
is the computation of the Monge-Kantorovich, or Wasserstein distance of exponent $2$ between two Gaussian measures, for which an exact formula is known:
$$
\MKd(G_1,G_2)^2=|m_1-m_2|^2+\Tr(A_1)+\Tr(A_2)-2\Tr((\sqrt{A_1}A_2\sqrt{A_1})^{1/2})\,,
$$
where $G_1$ and $G_2$ are two Gaussian probability measures on $\bR^n$ with means $m_1$ and $m_2$, and with (nonsingular) covariance matrices $A_1$ and $A_2$. See Proposition 7 in \cite{GivensShortt}.

Unfortunately, there are not so many analogous examples for which $E_\eps$ can be computed explicitly. By comparison, the Monge-Kantorovich, or Wasserstein distance $\MKd$ is better understood than $E_\eps$. However, if it may be hard 
to \textit{compute} explicitly $E_\eps$, it is relatively easy to compare $E_\eps$ to better known quantities. We shall discuss two such comparison methods below.

\subsubsection{Wigner and Husimi Transforms and Lower Bound for $E_\eps$}


To each density operator $R$, with integral kernel denoted $r\equiv r(x,y)$, one can associate a phase space function, called its \textit{Wigner transform}, defined as follows:
$$
W_\eps[R](x,\xi):=\tfrac1{(2\pi)^d}\int_{\bR^d}e^{-i\xi\cdot y}r(x+\tfrac12\eps y,x-\tfrac12\eps y)dy\in\bR\,.
$$
One easily checks that
$$
\iint_{\bR^d\times\bR^d}W_\eps[R](x,\xi)dxd\xi=\Tr(R)=1\,.
$$
This suggests the idea of thinking of $W_\eps[R]$ as a distribution function in kinetic theory, in other words a probability density in phase space. Unfortunately, it is easy to find examples of density operators $R$ for which $W_\eps[R]$
is not a.e. nonnegative. For instance, set 
$$
\psi(x)=2^{1/2}\pi^{-1/4}xe^{-x^2/2}\,,\quad x\in\bR\,,
$$
and consider the density operator $R=|\psi\ra\la\psi|$. One easily checks that
$$
W_\eps[R](0,0)=\tfrac1{2\pi}\int_\bR\psi(\tfrac12\eps y)\overline{\psi(-\tfrac12\eps y)}dy=-\tfrac1{\pi\eps}\|\psi\|_{L^2(\bR)}^2<0\,.
$$

For this reason, it is convenient to replace the Wigner transform with a nonnegative variant thereof, known as the \textit{Husimi transform}, obtained in terms of the Wigner function by the formula
$$
\widetilde W_\eps[R]:=e^{\eps\Dlt_{x,\xi}/4}W_\eps[R]\ge 0\,.
$$
See \cite{LionsPaul} for a presentation of these notions.

\begin{Thm}[Lower bound for $E_\eps$]\lb{T-LBEeps} For all $f\in L^1\cap\cP_2(\bR^d\times\bR^d)$ and all density operator $R\in\cD_2(\fH)$, one has
$$
E_\eps(f,R)^2\ge\max\left(d\eps\,,\,\MKd(f,\widetilde W_\eps[R])^2-d\eps\right)\,.
$$
\end{Thm}

See Theorem 2.4 (2) in \cite{FGTPArma} for a proof of this inequality.

\smallskip
This inequality compares the somewhat mysterious quantity $E_\eps(f,R)$ with the better known Monge-Kantorovich, or Wasserstein distance of exponent $2$ between $f$ and the Husimi transform of $R$, up to an error of order $O(\eps^{1/2})$. 
The main interest in this inequality is that it holds in the greatest possible generality. In other words, there is no restriction on either $f$ or $R$ for this inequality to hold.

\smallskip
Conversely, it will be useful to have an upper bound for $E_\eps$ in terms of the Monge-Kantorovich, or Wasserstein distance. This involves the notion of \textit{T\"oplitz quantization}, or \textit{positive quantization}, which is a kind
of ``approximate inverse'' of the Husimi transform.

\subsubsection{T\"oplitz Quantization and Upper Bound for $E_\eps$}


First, we recall the notion of \textit{Schr\"odinger coherent state}, or \textit{wave packet} centered at the phase space point $(q,p)\in\bR^d\times\bR^d$:
$$
|q+ip,\eps\ra=(\pi\eps)^{-d/4}e^{-|x-q|^2/2\eps}e^{ip\cdot x/\eps}\,.
$$
This formula is easily seen to define a normalized wave function on $\bR^d$. 

Next, for each positive Borel measure on $\bC^d\simeq\bR^d\times\bR^d$, one defines the \textit{T\"oplitz operator} with symbol $\mu$ by the formula
$$
\Op^T_\eps(\mu):=\tfrac1{(2\pi\eps)^d}\int_{\bC^d}|z,\eps\ra\la z,\eps|\mu(dz)\ge 0\,.
$$
Of course, at this level of generality $\Op^T_\eps(\mu)$ is only defined as an unbounded operator on $\fH=L^2(\bR^d)$. However, one easily checks that
$$
\Op^T_\eps(1)=I_\fH
$$
(where $\mu=1$ designates the Lebesgue measure on $\bC^d\simeq\bR^d\times\bR^d$). Similarly
$$
(2\pi\eps)^d\mu\in\cP(\bR^d\times\bR^d)\implies\Op^T_\eps(\mu)\in\cD(\fH)\,.
$$
The interested reader is referred to Appendix B of \cite{FGMouPaul} for a more detailed discussion of these operators.

With this material, we arrive at the following upper bound on the classical-to-quantum pseudo-distance $E_\eps$.

\begin{Thm}[Upper bound for $E_\eps$ for T\"oplitz density operators]\lb{T-UBEeps} 
Let $f,\mu$ belong to $\cP_2(\bR^d\times\bR^d)$, with $f\in L^1(\bR^d\times\bR^d)$. Then the T\"oplitz operator
$$
\Op^T_\eps((2\pi\hbar)^d\mu)\in\cD_2(\fH)\,,
$$
and
$$
E_\eps(f,\Op^T_\eps((2\pi\eps)^d\mu))^2\le\MKd(f,\mu)^2+d\eps\,.
$$
\end{Thm}

See Theorem 2.4 (1) of \cite{FGTPArma} for a proof of this upper bound.

Notice the difference between Theorems \ref{T-UBEeps} and \ref{T-LBEeps}: in Theorem \ref{T-UBEeps}, the density operator has to be a T\"oplitz operator, while in Theorem \ref{T-LBEeps}, there is no restriction on the density operator.
This observation is crucial for the next section, and for the interest of the pseudo-distance $E_\eps$ in the joint mean-field and classical limit discussed in this lecture.

\subsection{From $N$-body von Neumann to Vlasov}


At this point, we have gathered together the mathematical tools to study the joint mean field and classical limit corresponding to the oblique arrow in the diagram presented at the beginning of this lecture. Assuming that the interaction
potential $V$ has the same regularity as in the Dobrushin inequality presented in Lecture 1, we shall obtain a quantitative estimate for this limit in terms of the classical-to-quantum pseudo-distance $E_\eps$. It is interesting to observe
that this regularity assumption on $V$ is the exactly same as the one used to define the classical $N$-particle dynamics by means of the Cauchy-Lipschitz theorem. In some sense, this assumption could be thought of as minimal in order 
for this approach to the joint mean field and classical limit to be possible.

\begin{Thm} \lb{T-MF+CL} \lb{FGTPArma} Assume that the interaction potential $V$ satisfies assumptions (H1)-(H2) of Lecture 1. Let $R_{\eps,N}(t)=e^{-\frac{it\cH_N}\eps}R^{in}_{\eps,N}e^{+\frac{it\cH_N}\eps}$, where 
$$
\cH_N:=\sum_{j=1}^N-\tfrac12\eps^2\Dlt_{x_j}+\tfrac1N\sum_{1\le m<n\le N}V(x_m-x_n)\,,
$$
and $R^{in}_{\eps,N}\in\cD^s_2(\fH_N)$. On the other hand, let $f$ be the solution of the Vlasov equation
$$
(\d_t+\xi\cdot\grad_x)f(t,x,\xi)=\grad_xV_f(t,x)\cdot\grad_\xi f(t,x,\xi)\,,\qquad x,\xi\in\bR^d\,,\,\,t>0\,,
$$
where
$$
V_f(t,x)=\iint_{\bR^d\times\bR^d}V(x-y)f(t,y,\eta)dyd\eta
$$
is the mean field potential, with initial data $f^{in}\in L^1\cap\cP_2(\bR^d\times\bR^d)$. Set
$$
\Gamma:=1+2\max(1,2\Lip(\grad V)^2)\,.
$$

\smallskip
\noindent
(1) Then, for each $t\ge 0$ one has
$$
E_\eps(f(t),R_{\eps, N:1}(t))^2\!\le\!\frac{E_\eps((f^{in})^{\otimes N},R^{in}_{\eps,N})^2}Ne^{\Gamma t}\!+\!\frac{(2\|\grad V\|_{L^\infty})^2}{N-1}\frac{e^{\Gamma t}\!-\!1}{\Gamma}\,.
$$
(2) If moreover $R_{\eps,N}^{in}=\Op^T_\eps[(2\pi\eps)^{dN}(f^{in})^{\otimes N}]$, then
$$
\MKd(f(t),\widetilde W_\eps[R_{\eps,N:1}(t)])^2\le d\eps(1+e^{\Gamma t})+\frac{(2\|\grad V\|_{L^\infty})^2}{N-1}\frac{e^{\Gamma t}-1}{\Gamma}\,.
$$
\end{Thm}

\begin{proof}[Sketch of the proof of Theorem \ref{T-MF+CL}]
First we notice that (2) follows from (1) by a straightforward application of Theorems \ref{T-LBEeps} and \ref{T-UBEeps}. That the lower bound in Theorem \ref{T-LBEeps} applies to the most general finite energy densities is of utmost 
importance for this argument, since virtually nothing is known on $R_{\eps,N}(t)$, except that it is a density operator. On the other hand, that the upper bound in Theorem \ref{T-UBEeps} applies only to T\"oplitz densities is much less 
problematic, since it is used on the initial data, which can be chosen accordingly. 

It remains to prove (1). The idea is to follow the pattern outlined in the proof of the Dobrushin inequality in Lecture 1. 

Starting from a coupling $Q^{in}_{\eps,N}\in\cC((f^{in})^{\otimes N},R_{\eps,N}^{in})$, define $Q_{\eps,N}(t,X_N,\Xi_N)$to be the solution of the Cauchy problem
$$
\ba
\d_tQ_{\eps,N}(t,X_N,\Xi_N)+\sum_{j=1}^N(\xi_j\cdot\grad_{x_j}-\grad_xV_f(t,x_j)\cdot\grad_{\xi_j})Q_{\eps,N}(t,X_N,\Xi_N)
\\
+\tfrac{i}{\hb}[\cH_N,Q_{\eps,N}(t,X_N,\Xi_N)]&=0
\ea
$$
with initial data
$$
Q_{\eps,N}(0,X_N,\Xi_N)=Q^{in}_{\eps,N}(X_N,\Xi_N)\,.
$$
One easily checks that
$$
Q_{\eps,N}(t,\cdot,\cdot)\in\cC(f(t,\cdot,\cdot)^{\otimes N},R_{\eps,N}(t))\,,\qquad t\ge 0\,.
$$
To $Q_{\eps,N}(t,X_N,\Xi_N)$, we associate the function 
$$
\cD(t)\!=\!\frac1N\!\int_{\bR^{2dN}}\!\Tr_{\fH_N}(Q_{\eps,N}(t,X_N,\Xi_N)^\frac12c_\hb(X_N,\Xi_N)Q_{\eps,N}(t,X_N,\Xi_N)^\frac12)dX_Nd\Xi_N
$$
and observe first that, by definition of $E_\eps$,
$$
\frac{E_\eps((f(t,\cdot,\cdot))^{\otimes N},R_{\eps,N}(t))^2}N\le\cD(t)\,.
$$

On the other hand, since $R_{\eps,N}^{in}\in\cD^s(\fH_N)$, one easily checks that
$$
\cD(t)\!=\!\int_{\bR^{2dN}}\!\Tr_{\fH_N}(Q_{\eps,N}(t,X_N,\Xi_N)^\frac12J_kc_\hb(x_k,\xi_k)Q_{\eps,N}(t,X_N,\Xi_N)^\frac12)dX_Nd\Xi_N
$$
for each $k=1,\ldots,N$, where we recall that
$$
J_kA:=I_{\fH}^{\otimes(k-1)}\otimes A\otimes I_\fH^{\otimes(N-k)}\,.
$$
As a consequence
$$
E_\eps(f(t,\cdot,\cdot),R_{\eps,N:1}(t))^2\le\frac{E_\eps((f(t,\cdot,\cdot))^{\otimes N},R_{\eps,N}(t))^2}N\,.
$$

Finally, it remains to estimate $\cD(t)$. A first observation is that
$$
\ba
\frac{d\cD}{dt}\!=\!\tfrac1N\!\!\int_{\bR^{2dN}}\!\!\Tr_{\fH_N}(Q_{\eps,N}(t,X_N,\Xi_N)^\frac12d_\hb(t,X_N,\Xi_N)Q_{\eps,N}(t,X_N,\Xi_N)^\frac12)dX_Nd\Xi_N
\ea
$$
where
$$
d_\hb(t,X_N,\Xi_N)=\sum_{j=1}^N(\xi_j\cdot\grad_{x_j}-\grad_xV_f(t,x_j)\cdot\grad_{\xi_j})c_\hb(X_N,\Xi_N)+\tfrac{i}{\hb}[\cH_N,c_\hb(X_N,\Xi_N)]\,.
$$
(To see this, derive $\cD(t)$ under the integral sign and the trace, use the equation satisfied by $Q_{\eps,N}(t,X_N,\Xi_N)$, integrate by parts in all variables, and use the cyclicity of the trace.)

One easily computes
$$
\ba
d_\hb(t,X_N,\Xi_N)=&\sum_{j=1}^N(\xi_j+i\eps\grad_{y_j})\vee(x_j-y_j)
\\
&+\sum_{j=1}^N(\xi_j+i\eps\grad_{y_j})\vee\left(\grad_xV_f(t,x_j)-\tfrac1N\sum_{k=1}^N\grad V(y_j-y_k)\right)
\\
=&\sum_{j=1}^N(\xi_j+i\eps\grad_{y_j})\vee(x_j-y_j)
\\
&+\sum_{j=1}^N(\xi_j+i\eps\grad_{y_j})\vee\tfrac1N\sum_{k=1}^N(\grad V(x_j-x_k)-\grad V(y_j-y_k))
\\
&+\sum_{j=1}^N(\xi_j+i\eps\grad_{y_j})\vee\left(\grad_xV_f(t,x_j)-\tfrac1N\sum_{k=1}^N\grad V(x_j-x_k)\right)\,,
\ea
$$
with the notation
$$
(A_1,\ldots,A_d)\vee(B_1,\ldots,B_d)=\sum_{n=1}^d(A_nB_n+B_nA_n)\,.
$$
If $A_n=A_n^*$ and $B_n=B_n^*$ for $n=1,\ldots,n$, one has
$$
(A_1,\ldots,A_d)\vee(B_1,\ldots,B_d)\le\sum_{n=1}^d(A_n^2+B_n^2)\,,
$$
so that, using the Jensen inequality
$$
\ba
d_\hb(t,X_N,\Xi_N)\le c_\hb(X_N,\Xi_N)
\\
+\sum_{j=1}^N\left((\xi_j+i\eps\grad_{y_j})^2+\tfrac1N\sum_{k=1}^N(\grad V(x_j-x_k)-\grad V(y_j-y_k))^2\right)
\\
+\sum_{j=1}^N\left((\xi_j+i\eps\grad_{y_j})^2+\left(\grad_xV_f(t,x_j)-\tfrac1N\sum_{k=1}^N\grad V(x_j-x_k)\right)^2\right)
\\
\le c_\hb(X_N,\Xi_N)
\\
+\sum_{j=1}^N\left((\xi_j+i\eps\grad_{y_j})^2+\tfrac{2\Lip(\grad V)^2}N\sum_{k=1}^N(|x_j-y_j|^2+|x_k-y_k|^2)\right)
\\
+\sum_{j=1}^N\left((\xi_j+i\eps\grad_{y_j})^2+\left(\grad_xV_f(t,x_j)-\tfrac1N\sum_{k=1}^N\grad V(x_j-x_k)\right)^2\right)&\,.
\ea
$$
Hence
$$
\ba
d_\hb(t,X_N,\Xi_N)\le c_\hb(X_N,\Xi_N)+\sum_{j=1}^N\left(2(\xi_j+i\eps\grad_{y_j})^2+4\Lip(\grad V)^2|x_j-y_j|^2\right)
\\
+\sum_{j=1}^N\left(\grad_xV_f(t,x_j)-\tfrac1N\sum_{k=1}^N\grad V(x_j-x_k)\right)^2
\\
\le(1+2\max(1,2\Lip(\grad V)^2)c_\hb(X_N,\Xi_N)
\\
+\sum_{j=1}^N\left(\grad_xV_f(t,x_j)-\tfrac1N\sum_{k=1}^N\grad V(x_j-x_k)\right)^2&\,.
\ea
$$
Thus
$$
\ba
\frac{d\cD}{dt}(t)\le(1+2\max(1,2\Lip(\grad V)^2)\cD(t)
\\
+\int_{\bR^{2dN}}\!\tfrac1N\sum_{j=1}^N\left(\tfrac1N\sum_{k=1}^N(\grad_xV_f(t,x_j)-\grad V(x_j-x_k))\right)^2\prod_{k=1}^N\rho_f(t,x_j)dx_j&\,,
\ea
$$
with the notation
$$
\rho_f(t,x):=\int_{\bR^d}f(t,x,\xi)d\xi\,.
$$
Observe that, for $k\not=l=1,\ldots,N$, one has
$$
\ba
\int_{\bR^{2dN}}(\grad_xV_f(t,x_j)\!-\!\grad V(x_j-x_k))\cdot(\grad_xV_f(t,x_j)\!-\!\grad V(x_j-x_l))\prod_{k=1}^N\rho_f(t,x_j)dx_j\!=\!0
\ea
$$
by definition of $V_f$. Hence
$$
\ba
\frac{d\cD}{dt}(t)\le(1+2\max(1,2\Lip(\grad V)^2)\cD(t)
\\
+\int_{\bR^{2dN}}\!\tfrac1{N^3}\sum_{j,k=1}^N(\grad_xV_f(t,x_j)-\grad V(x_j-x_k))^2\prod_{k=1}^N\rho_f(t,x_j)dx_j
\\
\le(1+2\max(1,2\Lip(\grad V)^2)\cD(t)+\frac{4\|\grad V\|_{L^\infty}^2}{N}&\,.
\ea
$$
By Gronwall's inequality
$$
\frac{E_\eps((f(t,\cdot,\cdot))^{\otimes N},R_{\eps,N}(t))^2}N\le\cD(t)\le e^{\Ga t}\cD(0)+\frac{e^{\Ga t}-1}{\Ga}\frac{4\|\grad V\|_{L^\infty}^2}{N}\,.
$$
The term $\cD(0)$ in the right-hand side of the inequality above involves the initial coupling $Q^{in}_{\eps,N}\in\cC((f^{in})^{\otimes N},R_{\eps,N}^{in})$. Minimizing $\cD(0)$ in $Q^{in}_{\eps,N}$ implies that
$$
\frac{E_\eps((f(t,\cdot,\cdot))^{\otimes N},R_{\eps,N}(t))^2}N\le\frac{E_\eps((f^{in})^{\otimes N},R_{\eps,N}^{in})^2}Ne^{\Ga t}+\frac{e^{\Ga t}-1}{\Ga}\frac{4\|\grad V\|_{L^\infty}^2}{N}\,.
$$
The interested reader is invited to complete the missing details after reading the complete proof of Theorem 2.6 in \cite{FGTPArma}.
\end{proof}

\subsection{Mean-Field and Classical Limits: Quantum Klimontovich Solutions}


The result presented in the previous section, i.e. Theorem \ref{T-MF+CL} is very satisfying because it justifies rigorously the joint mean-field and classical limit as $\tfrac1N+\eps\to 0$ without any restriction on the rate at which $\tfrac1N$
and $\eps$ tend to $0$. For instance, we do not assume any distinguished limit (such as $\eps=N^{-1/3}$ for instance). There are however two shortcomings with this approach

\smallskip
\noindent
(1) the interaction force field  $-\grad V$ must be bounded and Lipschitz continuous, and

\noindent
(2) the best convergence rate as $\tfrac1N+\eps\to 0$ is achieved provided that the initial $N$-particle density operator $R_{\eps,N}^{in}$ is of the form $\Op^T_\eps[(2\pi\eps)^{dN}(f^{in})^{\otimes N}]$.

\smallskip
There are serious difficulties in removing the restriction mentioned in (1); this will be discussed in the next section. The restriction (2) on the initial data, is less formidable. In fact, this can be done by using the formalism of quantum
Klimontovich solutions introduced in lecture 2. 

The idea is again to start from the equation
$$
i\eps\d_t(\cM_N(t)-\cR_{\psi_\eps}(t))(A(t))=-\cC[V,\cM_N(t)-\cR_{\psi_\eps}(t),\cM_N(t)](A(t))
$$
where $\psi_\eps$ is the solution of the Hartree equation with semiclassical scaling
\be\lb{HartreeSemicl}
i\eps\d_t\psi_\eps=-\tfrac12\eps^2\Dlt_x\psi_\eps+V\star_x|\psi_\eps|^2\psi_\eps\,,\qquad\psi_\eps\rstr_{t=0}=\psi_\eps^{in}\,.
\ee
We recall that $\cM_N(t)$ is the quantum Klimontovich solution, while $\cR_{\psi_\eps}(t)$ is the ``chaotic morphism'' associated to the Hartree solution $\psi_\eps$ by Definition \ref{D-ChaoMor}.

In the proof of Theorem \ref{T-MFh=1}, we have used the rather naive estimate
$$
\frac1\eps\|\cM_N(t)([E_\om,A(t)])\|\le\frac1\eps\|[E_\om,A(t)]\|\le\frac2\eps\|A(t)\|\,.
$$
This estimate is clearly suboptimal if $A(t)$ is a multiplication operator, since in that case $[E_\om,A(t)]=0$. In general, one should try to use whatever cancellations might appear in the commutator $[E_\om,A(t)]$ in order to offset the
growth caused by the $1/\eps$ factor. The key idea is to restrict the class of time-dependent operators $A(t)$ used in this estimate --- specifically, one takes for $A(t)$ Weyl pseudo-differential operators\footnote{For $a\equiv a(x,\xi)$
belonging to $\cS(\bR^d\times\bR^d)$, one defines the Weyl operator with symbol $a$ by the duality formula
$$
\la\psi|\Op^W_\eps[a]|\phi\ra:=\iint_{\bR^d\times\bR^d}W_\eps[|\phi\ra\la\psi|](x,\xi)a(x,\xi)dxd\xi\,.
$$
This duality formula can be extended to the case where $a\in\cS'(\bR^d\times\bR^d)$ and defines $\Op^W_\eps[a]$ as a continuous linear map from $\cS(\bR^d)$ to $\cS'(\bR^d)$. One has $\Op^W_\eps[a]^*=\Op^W_\eps[\overline a]$.
The Calder\'on-Vaillancourt theorem states that
$$
\|\Op^W_\eps[a]\|_{\cL(L^2(\bR^d))}\le\g_d\max_{|\a|,|\b|\le[d/2]+1}\|\d_x^\a\d_\xi^\b a(x,\xi)\|_{L^\infty(\bR^d\times\bR^d)}
$$
for some constant $\g_d>0$ depending only on the space dimension $d$ \cite{Boulke}. One has
$$
[x_m,\Op_\eps^W[a]]=i\eps\Op_\eps^W[\d_{\xi_m}a]\,,\quad[-i\eps\d_{x_m},\Op_\eps^W[a]]=-i\eps\Op_\eps^W[\d_{x_m}a]\,,\quad m=1,\ldots,d\,.
$$} 
conjugated by the dynamics of the time-dependent Schr\"odinger equation with time-dependent mean-field potential $V\star_x|\psi_\eps(t,\cdot)|^2$. The important estimate is to be found in Lemma 4.1 in \cite{FGTPEmpirical}; is consists 
of a bound for the quantity
$$
\mathcal N(t):=\max_{1\le j\le d}(\|[x_j,A(t)]\|+\|[-i\eps\d_{x_j},A(t)]\|)\,.
$$
One can check that
$$
\cN(t)\le\cN(\tau)+\max(1,\Ga_2)\left|\int_\tau^t\cN(s)ds\right|\,,
$$
where
$$
\Ga_2:=\max_{1\le j\le d}\sum_{k=1}^d\tfrac1{(2\pi)^d}\int_{\bR^d}|\om_j||\om_k||\hat V(\om)|d\om\,,
$$
and conclude by Gronwall's inequality. That such an estimate helps in controlling the term $\frac1\eps\|[E_\om,A(t)]\|$ follows from the elementary formula
$$
[E_\om,A(s)]=\int_0^1E_{\l\om}[i\om\cdot x,A(s)]E_{(1-\l)\om}d\l\,.
$$
Observe that if $A(s)$ is a Weyl operator, the integrand is of order $\eps$ in operator norm by the footnote above, provided that the symbol of $A(s)$ has sufficiently many bounded derivatives in $x$ and $\xi$. This bound in operator 
norm is then propagated by conjugation with the Hartree dynamics, which is a unitary operator. The interested reader is referred to the proof of Lemma 4.1 in \cite{FGTPEmpirical} for the missing details.

With this (fundamental) observation, one arrives at the following convergence rate for the joint mean-field and classical limit. We shall need the following notation
$$
\|f\|_{n,n,\infty}:=\max_{\max(|\a|,|\b|)\le n}\|\d_x^\a\d_x^\b f\|_{L^\infty(\bR^d_x\times\bR^d_\xi)}
$$
for all $f\in C^{n,n}_b(\bR^d_x\times\bR^d_\xi)$, the linear space of functions $f\equiv f(x,\xi)$ such that $\d_x^\a\d_x^\b f$ exists and is continuous and bounded on $\bR^d_x\times\bR^d_\xi$ for all multiindices $\a,\b$ with length 
at most $n$. We designate by $C^{n,n}_b(\bR^d_x\times\bR^d_\xi)'$ the topological dual of $C^{n,n}_b(\bR^d_x\times\bR^d_\xi)$ with the topology defined by the norm $\|\cdot\|_{n,n,\infty}$, and by $\|\cdot\|'_{n,n,\infty}$ the dual
norm. Specifically, for each linear functional $L\in C^{n,n}_b(\bR^d_x\times\bR^d_\xi)'$, one defines
$$
\|L\|'_{n,n,\infty}:=\sup_{\|f\|_{n,n,\infty}\le1}|\la L,f\ra|\,.
$$

\begin{Thm}\lb{T-MF+CLKlim}\cite{FGTPEmpirical}
Assume that $V\in C_0(\bR^d)$ satisfies (H1) and
$$
\bV:=\tfrac1{(2\pi)^d}\int_{\bR^d}(1+|\om|)^{[d/2]+3}|\hat V(\om)|d\om<\infty\,.
$$
Let $\cH_N$ be the scaled $N$-particle Hamiltonian \eqref{HNsemicl}, and set
$$
\Psi_{\eps,N}(t,\cdot):=e^{-it\cH_N/\eps}(\psi_\eps^{in})^{\otimes N}\,.
$$
Then, for each $\hb\in(0,1]$, each $N\ge 1$ and each $t\ge 0$
$$
\ba
\|W_\eps[R_{\eps,N:1}(t)]-W_\eps[|\psi_\eps(t,\cdot)\ra\la\psi_\eps(t,\cdot)|]\|'_{[d/2]+2,[d/2]+2,\infty}
\\
\le\frac{\g_d+1}{\sqrt{N}}\exp\left(\sqrt{d}\g_dte^{t\max(1,\Ga_2)}\bV\right)
\ea
$$
for some constant $\g_d$ depending only on the space dimension $d$.
\end{Thm}

\smallskip
The interested reader is referred to section 4 of \cite{FGTPEmpirical} for the detailed proof of this result, which is somewhat technical.

\subsection{From the Quantum Coulomb Gas to Pressureless Euler-Poisson}


So far in this lecture and in Lecture 1, we have dealt with regular (at least $C^{1,1}$) potential. The treatment of singular potentials at the end of Lecture 2 does not seem compatible with the classical limit --- observe indeed the presence
of the $1/\hb$ term in the exponential amplifying rate in Theorem \ref{T-Coulh=1} (3). In order to pass to the limit in the quantum dynamical equation as both $1/N$ and $\eps$ tend to $0$, we need an estimate on the Coulomb force field 
replacing the Lipschitz continuity argument used in the Dobrushin inequality. 

First we recall the $N$-particle quantum Hamiltonian for the Coulomb gas:
$$
\cH_N=\sum_{k=1}^N-\tfrac12\eps^2\Dlt_{x_k}+\tfrac1N\sum_{1\le k<l\le N}\frac1{|x_k-x_l|}\,.
$$
By Kato's theorem, the operator $\cH_N$ has an extension to $\fH_N=L^2(\bR^{3N})$ as an unbounded operator such that $\cH_N=\cH_N^*$.

Henceforth, we assume that the $N$-particle initial state is factorized, i.e. is of the form
$$
\Psi_{\eps,N}^{in}:=\prod_{j=1}^N\psi_\eps(x_j)\,,
$$
where $\psi_\eps$ is a normalized element of $\fH$. The $N$-particle wave function at time $t$ is therefore
$$
\Psi_{\eps,N}(t,\cdot)=e^{-it\cH_N/\eps}\Psi_{\eps,N}^{in}\,.
$$
(Indeed, $e^{-it\cH_N/\eps}$ is a unitary group on $\fH_N$ according to Stone's theorem.)

In addition, we shall assume that the Wigner function of the initial single-particle state satisfies
$$
W_\eps[|\psi_\eps^{in}\ra\la\psi_\eps^{in}|](x,\xi)\to\rho^{in}(x)\de(\xi-u^{in}(x))
$$
in $\cS'(\bR^3\times\bR^3)$ as $\eps\to 0$. This type of phase space probability measure is referred to as a ``monokinetic'' distribution.

\smallskip
\noindent
\textbf{Example.} Perhaps the most famous example of (single-particle) wave function leading to a ``monokinetic'' Wigner measure in the vanishing $\eps$ limit is the case of \textit{WKB wave function}:
$$
\psi_\eps^{in}(x)=a^{in}(x)e^{iS^{in}(x)/\eps}\,,
$$
with 
$$
\|a^{in}\|_{L^2(\bR^3)}=1\quad\text{ and }\quad S^{in}\in W^{1,\infty}(\bR^3)\,.
$$
One easily checks that
$$
W_\eps[|\psi_\eps^{in}\ra\la\psi_\eps^{in}|](x,\xi)\to a^{in}(x)^2\de(\xi-\grad S^{in}(x))\,.
$$

\subsubsection{The Pressureless Euler-Poisson System}


The target equation of interest here is the following pressureless Euler-Poisson system. Its unknown is $(\rho,u)$, where $\rho(t,x)\ge 0$ is the gas density and $u(t,x)\in\bR^3$ its velocity field. It reads
$$
\left\{
\ba
{}&\d_t\rho+\Div_x(\rho u)=0\,,&&\rho\rstr_{t=0}=\rho^{in}\,,
\\
&\d_tu+u\cdot\grad_xu=-\grad_x\tfrac1{|x|}\star_x\rho\,,\qquad&&u\rstr_{t=0}=u^{in}\,.
\ea
\right.
$$
The Euler-Poisson system is related to the Vlasov-Poisson system by the following observation: if $(\rho,u)$ is a classical solution of the pressureless Euler-Poisson system, the monokinetic phase space probability measure
$$
f(t,x,\xi):=\rho(t,x)\de(\xi-u(t,x))
$$
is a solution of the Vlasov-Poisson system
$$
\left\{\ba
{}&\d_tf+\xi\cdot\grad_xf-\grad_xV_f(t,x)\cdot\grad_\xi f=0\,,
\\
&-\Dlt_xV_f(t,x)=4\pi\int_{\bR^3}f(t,x,\xi)d\xi\,.
\ea
\right.
$$

One easily proves the following local existence result for the pressureless Euler-Poisson system.

\begin{Thm}[Local Existence/Uniqueness for Euler-Poisson]
Let $u^{in}\in L^\infty(\bR^3)$ be such that $\grad_xu^{in}\in H^{2m}(\bR^3)$, and let $\rho^{in}\in H^{2m}(\bR^3)$ satisfy
$$
\rho^{in}(x)\ge 0\text{ for a.e. }x\in\bR^3\,,\quad\text{ and }\quad\int_{\bR^3}\rho^{in}(y)dy=1\,.
$$
Then

\smallskip
\noindent
(1) there exists $T\equiv T[\|\rho^{in}\|_{H^{2m}(\bR^3)}+\|\grad_xu^{in}\|_{H^{2m}(\bR^3)}]>0$, and a unique solution $(\rho,u)$ of the Euler-Poisson system such that
$$
u\in L^\infty([0,T]\times\bR^3)\quad\text{ while }\rho\text{ and }\grad_xu\in C([0,T],H^{2m}(\bR^3))\,;
$$
(2) besides, for all $t\in[0,T]$, one has
$$
\rho(t,x)\ge 0\text{ for a.e. }x\in\bR^3\,,\quad\text{ and }\quad\int_{\bR^3}\rho(t,y)dy=1\,.
$$
\end{Thm}

This result is proved by a standard energy method (notice that $\grad^k_x\rho$ and $\grad^{k+1}_xu$ appear at the same order in this estimate; see \cite{SerfatyDMJ2020} for more details).

\subsubsection{From $N$-Body Schr\"odinger to the Euler-Poisson System}


Our result on the joint mean field and classical limit for the $N$-particle Coulomb gas (i.e. the oblique arrow in the diagram at the beginning of Lecture 3) with monokinetic initial data is summarized in the following theorem.

\begin{Thm}\cite{FGTPCpam}
Let $\rho^{in}\in H^4(\bR^3)\cap\cP(\bR^3)$ and $u^{in}\in L^\infty(\bR^3)^3$ be such that $\grad u^{in}\in H^4(\bR^3)^3$. 

Let $\Psi_{\eps,N}^{in}=(\psi_\eps^{in})^{\otimes N}$, with $\|\psi_\eps\|_{L^2}=1$ satisfying
$$
\sup_{0<\eps<1}\la\psi_\eps^{in}|\,\,|\eps D_x|^4\,|\psi_\eps^{in}\ra<\infty\,,\quad\la\psi_\eps^{in}|\,\,|\eps D_x-u^{in}|^2\,|\psi_\eps^{in}\ra\to 0
$$
and
$$
\iint_{\bR^3\times\bR^3}\frac{(|\psi_\eps^{in}(x)|^2-\rho^{in}(x))(|\psi_\eps^{in}(y)|^2-\rho^{in}(y))}{|x-y|}dxdy\to 0
$$
as $\eps\to 0$. Set
$$
\cH_N:=\sum_{k=1}^N-\tfrac12\eps^2\Dlt_{x_k}+\tfrac1N\sum_{1\le k<l\le N}\frac1{|x_k-x_l|}\,,
$$
the $N$-particle quantum Hamiltonian for the Coulomb gas. Let 
$$
\Psi_{\eps,N}(t,\cdot):=e^{\frac{it\cH_N}\eps}\Psi_{\eps,N}^{in}\,,\qquad R_{\eps,N}(t)=|\Psi_{\eps,N}(t,\cdot)\ra\la\Psi_{\eps,N}(t,\cdot)|\,,
$$
and let $R_{\eps,N:1}(t)$ be the first marginal (reduced density operator) of $R_{\eps,N}(t)$.

Let $(\rho,u)$ be the (classical) solution on $[0,T]\times\bR^3$ for some $T>0$ of the pressureless Euler-Poisson system with initial data $(\rho^{in},u^{in})$.

\medskip
Then, in the limit as $\eps+\tfrac1N\to 0$, the reduced density operator of the $N$-particle wave function $\Psi_{\eps,N}(t)$ satisfies
$$
W_\eps[R_{\eps,N:1}(t)](x,\xi)\to\rho(t,x)\de(\xi-u(t,x))\quad\text{ in }\cS'(\bR^3)\,,
$$
and
$$
\ba
\int_{\bR^3}W_\eps[R_{\eps,N:1}(t)]d\xi&\to\rho(t,\cdot)
\\
\int_{\bR^3}\xi W_\eps[R_{\eps,N:1}(t)]d\xi&\to\rho u(t,\cdot)
\ea
$$
for the narrow topology of Radon measures on $\bR^3$.
\end{Thm}

\smallskip
The conclusions of the theorem above can be recast as follows: in the limit as $\eps+\tfrac1N\to 0$, one has
$$
\ba
\int_{\bR^{3(N-1)}}|\Psi_{\eps,N}(t,\cdot,X_{2,N})|^2dX_{2,N}&\to\rho(t,\cdot)
\\
\hb\int_{\bR^{3(N-1)}}\Im\left(\overline{\Psi_{\eps,N}}\grad_{x_1}\Psi_{\eps,N}\right)(t,\cdot,X_{2,N})dX_{2,N}&\to\rho u(t,\cdot)
\ea
$$
for the narrow topology of Radon measures on $\bR^3$.

\smallskip
The key new ingredient used in the proof of this theorem is the following remarkable inequality, due to Serfaty \cite{SerfatyDMJ2020}.

For all $\rho\in L^\infty(\bR^3)$, all $u\in W^{1,\infty}(\bR^3)^3$ and all $X_N\in\bR^{3N}$, set
$$
\left\{
\ba
{}&F[X_N,\rho]:=\iint_{x\not=y}\frac{(\mu_{X_N}-\rho)(dx)(\mu_{X_N}-\rho)(dy)}{|x-y|}\,,
\\
&G[X_N,\rho,u]\!:=\!\!\iint_{x\not=y}\!\!\!\tfrac{(u(x)\!-\!u(y))\cdot(x\!-\!y)}{|x-y|^3}(\mu_{X_N}\!-\!\rho)(dx)(\mu_{X_N}\!-\!\rho)(dy)\,.
\ea
\right.
$$

\bigskip
\noindent
\fbox{\sc Serfaty's Inequality}

\medskip
There exists $C>2$ such that, for all $\rho\in L^\infty(\bR^3)$, all $u\in W^{1,\infty}(\bR^3)^3$ and a.e. $X_N\in\bR^{3N}$
$$
|G[X_N,\rho,u]|\le C\|\grad u\|_{L^\infty}F_N[X_N,\rho]+\frac{C}{N^{1/3}}(1+\|\rho\|_{L^\infty})(1+\|u\|_{W^{1,\infty}})\,.
$$

Besides, there exists $C'>0$ such that
$$
F[X_N,\rho]\ge -\frac{C'}{N^{2/3}}(1+\|\rho\|_{L^\infty(\bR^3)})\,.
$$

\bigskip
With this inequality, Serfaty and Duerinckx proved that, if $t\mapsto(X_N,\Xi_N)(t)$ is a solution of the $N$-body Newton equations with Coulomb repulsive potential such that
$$
\tfrac1N\sum_{j=1}^N\de_{x_j(0)}\to\rho^{in}\quad\text{ and }\quad\frac1N\sum_{j=1}^N\xi_j(0)\de_{x_j(0)}\to\rho^{in} u^{in}
$$
in the narrow topology as $N\to\infty$, then, for each $t\in[0,T]$, one has
$$
\mu_{X_N(t)}\!:=\!\tfrac1N\sum_{j=1}^N\de_{x_j(t)}\!\to\!\rho(t,\cdot)\quad\text{ and }\quad\tfrac1N\sum_{j=1}^N\xi_j(t)\de_{x_j(t)}\!\to\!\rho u(t,\cdot)\,.
$$
To arrive at this result, Serfaty and Duerinckx obtained a Gronwall type inequality for the classical modulated energy for Klimontovich solutions of the Vlasov equation, in other words, for the quantity
$$
\tfrac1N\sum_{j=1}^N|\xi_j(t)-u(t,x_j(t))|^2+F[X_N(t),\rho(t,\cdot)]
$$

Serfaty's inequality is used in a slightly different manner in the proof of the above theorem.

\begin{proof}[Sketch of the proof]
First, we define some appropriate modulation of the total energy of the quantum particles. With the notation 
$$
J_1A=A\otimes \overbrace{I\otimes\ldots\otimes I}^{N-1\text{ terms}}\,,
$$
we consider the quantity (modulated energy)
$$
\ba
\cE[\Psi_{\hbar,N},\rho,u](t):=&\left\la\Psi_{\hbar,N}(t)\big|\,J_1|-i\hbar\grad_x-u(t,\cdot)|^2\,\big|\Psi_{\hbar,N}(t)\right\ra
\\	\\
&+\la\Psi_{\hbar,N}(t)|\,F[X_N,\rho(t,\cdot)]\,|\Psi_{\hbar,N}(t)\ra\,.
\ea
$$
Denoting $\Sigma:=\tfrac12(\grad_xu+(\grad_xu)^T)$ the deformation tensor, some fastidious (but easy!) computations show that 
$$
\ba
\frac{d}{dt}\cE[\Psi_{\hbar,N},\rho,u](t)+2\left\la\Psi_{\hbar,N}\big|\,J_1((i\hbar\grad_x\!+\!u)^T\Sigma(i\hbar\grad_x\!+\!u))\,\big|\Psi_{\hbar,N}\right\ra&
\\	\\
=\tfrac12\hb^2\left\la\Psi_{\hbar,N}\big|\,J_1(\Dlt_x\Div_xu(t,\cdot))\,\big|\Psi_{\hbar,N}\right\ra+\left\la\Psi_{\hbar,N}\big|\,G[X_N,\rho,u]\,\big|\Psi_{\hbar,N}\right\ra&\,.
\ea
$$
Using Gronwall's and Serfaty's inequalities, one arrives at the bound
$$
\ba
0\le\cE[\Psi_{\hbar,N},\rho,u](t)+\frac{C'}{N^{2/3}}(1+\|\rho\|_{L^\infty(\bR^3)})
\\
\le e^{CT\|\grad u\|_{L^\infty}}\left(\underbrace{\cE[\Psi_{\hbar,N},\rho,u](0)}_{\to 0}+\frac{C'}{N^{2/3}}(1+\|\rho\|_{L^\infty(\bR^3)})\right)
\\
+Te^{CT\|\grad u\|_{L^\infty}}\frac{C}{N^{1/3}}(1+\|\rho\|_{L^\infty})(1+\|u\|_{W^{1,\infty}})
\\
+Te^{CT\|\grad u\|_{L^\infty}}\tfrac12\hb^2\|\Dlt_x\Div_xu\|_{L^\infty}&\,.
\ea
$$

By the lower bound in Serfaty's inequality and the Cauchy-Schwarz inequality
$$
\ba
\cE[\Psi_{\hbar,N},\rho,u](t)&+\frac{C'}{N^{2/3}}(1+\|\rho\|_{L^\infty(\bR^3)})
\\
&\ge\left|\left\la\Psi_{\hbar,N}(t)\big|\,J_1(-i\hbar\grad_x-u(t,\cdot))\,\big|\Psi_{\hbar,N}(t)\right\ra\right|^2\,.
\ea
$$
This implies the announced convergence to the momentum density.

The narrow convergence of the densities is the second important conclusion deduced from the  modulated energy. It is specifically based on the properties of the potential energy. Starting from the decomposition (left to the reader as an exercise)
$$
\frac1{4\pi|x-y|}\!=\!\!\int_0^\infty\!\!\!dr\!\int_{\bR^3}\!G_r(x\!-\!z)G_r(y\!-\!z)dz\,,\qquad\text{ with }\quad G_r(w)\!:=\tfrac{e^{-\frac{|w|^2}{2r}}}{(2\pi r)^{\frac32}}\,,
$$
one can prove that
$$
\ba
\int_\eps^\infty\|e^{r\Dlt/2}\rho_{\hb,N:1}(t,\cdot)-\rho(t,\cdot)\|^2_{L^2}dr
\\
\le\underbrace{\la\Psi_{\hbar,N}(t)|\,F[X_N,\rho(t,\cdot)]\,|\Psi_{\hbar,N}(t)\ra}_{\to 0}+O(N^{-2/3})&\,,
\ea
$$
where
$$
\rho_{\hb,N:1}(t,\cdot):=\int_{\bR^{3(N-1)}}|\Psi_{\hb,N}(t,\cdot,X_{2,N})|^2dX_{2,N}\,.
$$
\end{proof}

The interested reader is referred to \cite{FGTPCpam} for the missing details.

\subsection{Miscellaneous Remarks}


\smallskip
\noindent
(1) As already mentioned in Remark (7) at the end of lecture 2, the mean-field limit for a gas of $N$ fermions with comparable kinetic and potential energies involve a distinguished limit which is reminiscent of a joint mean-field and classical
limit (with $\eps=N^{-1/3}$ in space dimension $3$). Hopefully, the material presented in lecture 3 might become useful to a better understanding of this case, which is of considerable importance, for instance in chemistry. The interested
reader should read \cite{NarnSewell} --- see also \cite{BPSS}.

\smallskip
\noindent
(2) The method of \cite{FGTPCpam} based on Serfaty's inequality can be used to derive rigorously the Vlasov-Poisson system from the Hartree equation in the Coulomb case and in the monokinetic setting: see Proposition 2.4 in \cite{FGTPCpam}.
 (This is the right vertical arrow in the diagram at the beginning of lecture 3, in other words, the classical limit of the Hartree equation leading to the Vlasov-Poisson system). This problem has already been treated some time ago: see Theorem IV.5 
in \cite{LionsPaul}. (See also \cite{BPSS} in the case of regular potentials.) There is however a fundamental difference between Proposition 2.4 in \cite{FGTPCpam} and Theorem IV.5 in \cite{LionsPaul}. Indeed, Theorem IV.5 in \cite{LionsPaul} 
assumes that the Wigner transform of the states considered is bounded in $L^\infty([0,T],L^2(\bR^3\times\bR^3))$. This incompatible with the monokinetic setting in \cite{FGTPCpam}, where the Wigner functions considered converge to a Dirac 
distribution in the momentum variable. Thus Proposition 2.4 in \cite{FGTPCpam} and Theorem IV.5 in \cite{LionsPaul} both establish the validity of the classical limit of the Hartree equation, but in radically different asymptotic regimes.

\smallskip
\noindent
(3) All the quantum dynamics considered here do not include any magnetic field. The quantum-to-classical Wasserstein pseudo distance can also be used in the presence of an external magnetic field: see \cite{IBPPhD}.

\smallskip
\noindent
(4) Whether the results in \cite{SerfatyDMJ2020} or in \cite{FGTPCpam} can be extended beyond the monokinetic case is a major open question. What is at stake is a rigorous derivation of the Vlasov-Poisson system starting from a classical
or a quantum Coulomb gas, a notoriously difficult and fundamental problem in the kinetic theory of charged particles.



\end{document}